\documentclass[journal]{IEEEtran}
\usepackage{graphicx}
\usepackage{booktabs}
\usepackage{cite}
\usepackage{algorithm}
\usepackage{algpseudocode}

\usepackage{amsmath,amssymb,amsfonts,bm}
\usepackage[caption=false,font=footnotesize]{subfig}
\usepackage{textcomp}
\usepackage{multirow}
\usepackage{tikz}
\usetikzlibrary{arrows.meta,positioning}

\newtheorem{proposition}{Proposition}
\newtheorem{theorem}{Theorem}

\newtheorem{remark}{Remark}
\newtheorem{lemma}{Lemma}

\newcommand{\pmse}[1]{{\scriptsize$\,\pm#1$}}

\long\def\comment#1{}

\newfont{\bbb}{msbm10 scaled 700}

\newfont{\bb}{msbm10 scaled 1100}

\begin{document}

\title{When to Move: Online Repositioning of Movable Antenna Arrays Under a Long-Term Actuation Energy Budget}

\author{Yunseob Tae,~\IEEEmembership{Graduate Student Member,~IEEE,} Jeongjae Lee,~\IEEEmembership{Member,~IEEE,} and~Songnam~Hong,~\IEEEmembership{Senior Member,~IEEE}
        \thanks{Y. Tae and S. Hong are with the Department of Electronic Engineering, Hanyang University, Seoul, Korea. (e-mail: \{tys7524, snhong\}@hanyang.ac.kr).}

        \thanks{J. Lee is with the Department of Electronic Engineering, Hanyang University, Seoul, Korea, and the Ming Hsieh Department of Electrical Engineering, University of Southern California, Los Angeles, CA, USA. (e-mail: jl\_939@usc.edu).}
        
        \thanks{This work was supported by the NRF (No. RS-2024-00409492) and the IITP
        (IITP-2026-RS-2023-00253914), funded by the Korea government (MSIT).}
        \thanks{Parts of this work were presented at IEEE VTC2026-Fall.}
}

\maketitle

\begin{abstract}
Movable antenna (MA) arrays improve multiuser throughput by repositioning their elements, but mechanical movement reduces transmission time and consumes actuation energy. Changes in user configurations create nonuniform demand for repositioning over time, while each movement determines which positions are reachable in subsequent slots. Thus, online control must decide when to spend a limited actuation-energy budget, not only where to position the antennas. We formulate causal effective-throughput maximization under displacement-dependent movement delay, inter-slot reachability, and a long-term average actuation-energy constraint. We propose Lyapunov-based online movable-antenna repositioning (LOMA), which adjusts the price of movement through a virtual queue and solves the resulting nonconvex position problem via proximal minorization--maximization, accounting for nonsmooth movement delay and promoting sparse antenna updates. We establish monotone ascent of the smoothed per-slot objective, a finite-horizon sample-path energy bound, almost-sure long-term budget feasibility under a Gaussian fading model with bounded large-scale gains, and a state-matched throughput comparison separating temporal energy allocation from per-slot optimization error. Simulations show that LOMA concentrates actuation after intermittent user changes and improves effective throughput over optimized fixed arrays and a per-slot-budget counterpart, while continuous channel drift favors more frequent incremental repositioning. Additional simulations with pilot-based channel estimates in sparse millimeter-wave channels show that LOMA retains most of the throughput benefits of online repositioning while operating near its prescribed long-term actuation-energy budget.
\end{abstract}


\section{Introduction}

Movable antenna (MA) arrays reconfigure the propagation channel by adjusting the positions of their antenna elements without adding radio-frequency chains~\cite{zhu2023movable}. This spatial flexibility can improve multiuser beamforming and effective throughput~\cite{feng2024weighted,ma2024multi}. Unlike electronic beamforming, however, physical repositioning incurs both movement delay and mechanical energy expenditure~\cite{ning2025movable}. These costs become important when an access point must serve changing user configurations with a limited actuation-energy supply.

With multiple actuators moving in parallel, movement delay is determined by the \emph{largest} element displacement, whereas mechanical energy accumulates over \emph{all} displacements. Each repositioning decision therefore affects the achievable channel, the transmission time available within the slot, and the actuation energy consumed over time. Under a long-term average actuation-energy constraint, the controller must determine not only where to reposition the antennas but also when the resulting throughput improvement justifies the movement cost.

When channels vary continuously, frequent small adjustments can effectively track changes in the preferred antenna geometry, making a per-slot energy allocation a reasonable choice. However, repositioning demand is not always distributed uniformly over time. A change in the scheduled user set can abruptly alter the preferred array geometry, making movement particularly valuable after such events, while relatively little movement may be needed between them. A fixed per-slot energy cap cannot transfer unused allowance to later reconfigurations, whereas a myopic controller without energy pricing may overspend. Allocating actuation energy across slots can therefore improve throughput by conserving energy when movement offers little benefit and spending it when repositioning becomes valuable. However, future user and channel changes are unknown, so these decisions must be made causally. The problem is further complicated by inter-slot reachability, since the current antenna positions determine which positions can be reached in the next slot.

\subsection{Related Work}\label{subsec:related}

\emph{Movement delay and energy.} Movement-aware MA designs have minimized the travel time to prescribed antenna positions~\cite{li2025trajectory} and optimized throughput while accounting for repositioning delay within a transmission block~\cite{wang2025throughput,tae2026vtc}. Other studies have modeled mechanical actuation energy to improve the energy efficiency of MA systems~\cite{wei2025mechanical,wei2026energy}, including multiuser uplink designs that jointly account for movement delay and energy consumption~\cite{ding2025energy}. More recently, an MA relay design has incorporated inter-slot displacement constraints and a finite-horizon budget on the total movement and transmission energy, using reinforcement learning to optimize antenna positions and transmit power under quasi-static channels~\cite{zhang2026relay}. Unlike this setting, which reserves a fixed-duration subslot for antenna movement and assumes that channels vary only with antenna positions, we consider displacement-dependent transmission-time loss and time-varying repositioning demand.

\emph{Channel acquisition and adaptive MA control.} Channel acquisition in MA systems has been investigated through spatial channel reconstruction using measurements at selected antenna positions~\cite{xiao2024channel,xu2024fas}. Measurement and communication trajectories have also been jointly designed under a time constraint~\cite{zheng2026trajectory}. To reduce the overhead of frequent position updates, two-timescale approaches optimize MA geometry based on statistical channel information while updating beamforming more frequently~\cite{zheng2025two,liang2026twotimescale}. Related studies have jointly optimized antenna positioning, user association, and beamforming in MA-enabled heterogeneous networks~\cite{deng2026dualtimescale}. Movement-aware online learning methods, including MoveUCB and CoMoveUCB, use scalar-rate feedback to learn position-dependent rewards while accounting for movement costs~\cite{ch26move,ch25comove}. In contrast, our work focuses on the causal allocation of actuation energy across slots under a long-term average constraint, using currently available field-response channel state information (CSI) without knowing future user or channel changes.

\emph{Online resource allocation.} Lyapunov drift-plus-penalty control provides a general approach to long-term average resource constraints without requiring knowledge of the underlying state distribution~\cite{neely2010}. In MA repositioning, however, the feasible action set in each slot depends on the previously selected antenna positions. Consequently, a conventional throughput comparison against an independently evolving reference policy cannot be applied directly, since different policies generally follow different position trajectories and face different subsequent feasible sets. This motivates a state-matched throughput comparison that also accounts for the optimization residual arising from the nonconvex per-slot position problem.

Existing MA designs have primarily addressed \emph{where} antennas should be positioned, how movement overhead affects transmission, or how channel information can be acquired and learned. These advances motivate a distinct online resource-allocation question: \emph{when} should an MA array spend its limited actuation energy under a long-term average budget as the preferred geometry changes? Addressing this question requires causal decisions that jointly account for displacement-dependent transmission-time loss, energy consumption across slots, and position-dependent reachability, without knowledge of future user or channel changes.

\subsection{Our Contributions}
\label{subsec:contributions}

We investigate online effective-throughput maximization by allocating a long-term average actuation-energy budget across slots under movement delay and inter-slot reachability. Our main contributions are summarized as follows:
\begin{itemize}
  \item We formulate the causal allocation of actuation energy over time and propose Lyapunov-based online movable-antenna repositioning (LOMA). A virtual queue adaptively prices movement energy across slots, while a proximal minorization--maximization algorithm handles the nonconvex position problem through log-sum-exp smoothing of the maximum-displacement delay and promotes sparse antenna updates through a queue-weighted movement penalty.
  \item We establish monotone ascent and limiting stationarity of the smoothed per-slot solver, a finite-horizon sample-path energy bound valid for every channel realization, and almost-sure long-term budget feasibility under Gaussian fading with bounded large-scale gains. To account for the dependence of feasible positions on previous decisions, we derive a state-matched throughput comparison that separates the performance loss due to online energy management from the per-slot optimization residual.
  \item We introduce LOMA-PS, a per-slot-budget variant using the same position solver, to isolate the benefit of temporal energy allocation. Simulations demonstrate that LOMA achieves higher effective throughput under intermittent user changes, whereas LOMA-PS better tracks continuous channel drift through frequent incremental repositioning. Additional simulations with ordinary time-division duplex (TDD) pilots in sparse millimeter-wave (mmWave) channels show that LOMA retains most of the ideal-CSI repositioning gain without dedicated mechanical CSI-probing movements while operating near the prescribed energy budget.
\end{itemize}

\subsection{Organization and Notation}
\label{subsec:organization}

The remainder of this paper is organized as follows. Section~\ref{sec:system} introduces the system model. Section~\ref{sec:formulation} formulates the online throughput maximization problem under a long-term actuation-energy constraint. Section~\ref{sec:framework} develops the Lyapunov-based online control framework. Section~\ref{sec:algorithm} presents the proposed LOMA algorithm. Section~\ref{sec:analysis} provides the theoretical analysis, including convergence and performance guarantees. Section~\ref{sec:simulation} presents numerical results. Finally, Section~\ref{sec:conclusion} concludes the paper.

\textit{Notation:} Boldface lowercase and uppercase letters denote vectors and matrices, respectively, and $\mathbf{0}$ and $\mathbf{1}$ denote the all-zero and all-ones vectors of appropriate dimensions. $(\cdot)^{\mathrm{T}}$ and $(\cdot)^{\mathrm{H}}$ denote the transpose and conjugate transpose, respectively. $\|\cdot\|_{1}$, $\|\cdot\|_{2}$, and $\|\cdot\|_{\infty}$ denote the $\ell_{1}$, $\ell_{2}$, and $\ell_{\infty}$ norms, respectively, $|\cdot|$ denotes the absolute value of a scalar or the cardinality of a set, and $[\,\cdot\,]^{+}\triangleq\max\{\cdot,0\}$. $\mathbb{E}[\cdot]$ denotes the expectation, and $\mathcal{O}(\cdot)$ denotes the standard big-O notation. $\mathcal{N}(0,\sigma^{2})$ and $\mathcal{CN}(0,\sigma^{2})$ denote the real and circularly symmetric complex Gaussian distributions with zero mean and variance $\sigma^{2}$, respectively, and $\mathcal{U}[a,b]$ denotes the uniform distribution over $[a,b]$.

\section{System Model}\label{sec:system}

We consider a downlink MA-aided system, in which a base station~(BS) equipped with $M$ MAs, each driven by an independent mechanical actuator, serves $K$ single-antenna user equipments (UEs), with $M\ge K$ for zero-forcing (ZF) precoding. As illustrated in Fig.~\ref{fig:system}, the system operates over time slots indexed by $t\in\{1,2,\dots\}$, each of duration $T_s$, within which the channel is assumed to remain constant. Let $\tau_{\mathrm p}\in[0,T_s)$ denote an optional uplink-pilot interval and define $T_d\triangleq T_s-\tau_{\mathrm p}$ as the time available for repositioning and downlink data. The analytical ideal-CSI setting uses $\tau_{\mathrm p}=0$, whereas the pilot-based validation in Section~\ref{subsec:estimated_csi} uses a positive $\tau_{\mathrm p}$. Let $\mathbf{x}_t=[\,x_{t,1}\;x_{t,2}\;\cdots\;x_{t,M}\,]^{\mathrm{T}}\in\mathbb{R}^{M}$ denote the antenna position vector (APV) at slot $t$, and let $\mathbf{x}_0$ denote the initial APV. The MAs are indexed in ascending order of their positions, i.e., $0\le x_{t,1}<x_{t,2}<\cdots<x_{t,M}\le D$, where $D$ is the length of the linear region within which the MAs can be positioned.

At the beginning of each slot, the MAs are mechanically repositioned from $\mathbf{x}_{t-1}$ to $\mathbf{x}_t$. Data transmission is suspended during this repositioning because antenna motion induces Doppler shifts and rapid channel variations~\cite{wang2025throughput}, and the time it occupies is referred to as the movement delay. Since the MAs can move simultaneously at a common speed $v$, the movement delay of slot $t$ is determined by the \emph{maximum} displacement, i.e.,
\begin{equation}
  \tau^{\mathrm{mov}}_{t}\triangleq
  \frac{\|\mathbf{x}_t-\mathbf{x}_{t-1}\|_{\infty}}{v},
  \label{eq:mov_delay}
\end{equation}
and the effective data transmission time of slot $t$ is $T^{\mathrm{data}}_{t}=T_d-\tau^{\mathrm{mov}}_{t}$. As illustrated in Fig.~\ref{fig:system}, optional uplink pilots are transmitted after the MAs settle, so that downlink precoding uses CSI estimated at the same APV used for data transmission. In contrast, mechanical energy is consumed by each moving actuator and thus accumulates over \emph{all} antenna displacements, as modeled in Section~\ref{sec:formulation}. Consequently, repositioning incurs both actuation-energy consumption and a reduction in the time available for data transmission.

\begin{figure}[t]
\centering
\includegraphics[width=\columnwidth]{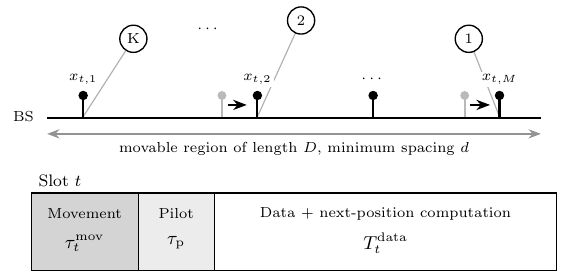}
\caption{Illustration of the BS-side linear MA array and the slot structure (not to scale). Antenna repositioning is followed by optional pilot-based CSI acquisition and downlink data transmission, during which the next position decision is computed.}
\label{fig:system}
\end{figure}

\subsection{Channel Model}\label{subsec:channel}

We adopt a far-field geometric channel model for the linear MA array. For a given APV $\mathbf{x}$ and angle of departure (AoD) $\theta\in[0,\pi]$, the array response vector is given by $\mathbf{a}(\mathbf{x},\theta)=[\,e^{\jmath\frac{2\pi}{\lambda}x_{1}\cos\theta}\;\cdots\;e^{\jmath\frac{2\pi}{\lambda}x_{M}\cos\theta}\,]^{\mathrm{T}}\in\mathbb{C}^{M}$, where $\lambda$ denotes the carrier wavelength. For sparse mmWave channels, we model the channel from the BS to UE $k$ as the sum of one dominant line-of-sight (LoS) path and $L_{k}$ non-line-of-sight (NLoS) paths, i.e.,
\begin{equation}
  \mathbf{h}_{k}\big(\mathbf{x};\bm{\xi}_{k}\big)
  =\alpha^{\mathrm{LoS}}_{k}\,
    \mathbf{a}\big(\mathbf{x},\theta^{\mathrm{LoS}}_{k}\big)
  +\sum_{\ell=1}^{L_{k}}\alpha^{\mathrm{NLoS}}_{k,\ell}\,
    \mathbf{a}\big(\mathbf{x},\theta^{\mathrm{NLoS}}_{k,\ell}\big),
  \label{eq:channel}
\end{equation}
where $\bm{\xi}_{k}\triangleq\{\alpha^{\mathrm{LoS}}_{k},\theta^{\mathrm{LoS}}_{k},\{\alpha^{\mathrm{NLoS}}_{k,\ell},\theta^{\mathrm{NLoS}}_{k,\ell}\}_{\ell=1}^{L_{k}}\}$ collects the complex path gains and AoDs associated with UE $k$. Let $\beta_{k}$ and $\kappa_{k}$ denote the large-scale path gain and Rician factor of UE $k$, respectively. The power fractions of the LoS path and each NLoS path are given by $\rho^{\mathrm{LoS}}_{k}=\kappa_{k}/(\kappa_{k}+1)$ and $\rho^{\mathrm{NLoS}}_{k}=1/\big((\kappa_{k}+1)L_{k}\big)$, respectively. The resulting $M\times K$ channel matrix is
\begin{equation}
  \mathbf{H}\big(\mathbf{x};\bm{\xi}\big)
  =\big[\,\mathbf{h}_{1}(\mathbf{x};\bm{\xi}_{1})\;\;
    \mathbf{h}_{2}(\mathbf{x};\bm{\xi}_{2})\;\;\cdots\;\;
    \mathbf{h}_{K}(\mathbf{x};\bm{\xi}_{K})\,\big],
  \label{eq:composite_channel}
\end{equation}
where $\bm{\xi}\triangleq\{\bm{\xi}_{k}\}_{k=1}^{K}$. Unlike conventional fixed-position antenna (FPA) systems, the channel in \eqref{eq:channel} explicitly depends on the APV, allowing channel reconfiguration through antenna repositioning. However, this position dependence also leads to a generally nonconvex achievable-rate optimization problem.

\subsection{Temporal Evolution}\label{subsec:temporal}

We consider two forms of temporal channel variation: continuous evolution caused by UE mobility and event-driven changes in the served user set. For continuous variation, we adopt a mobility-driven model in which UE motion determines the LoS geometry and phase evolution, while the NLoS gains follow a Doppler-correlated fading approximation.

\noindent\textbf{UE Mobility.}
The BS is located at the origin, with the movable region aligned with the horizontal axis. UE $k$ is located at $\mathbf{p}^{(t)}_{k}=[\,p^{(t)}_{k,1}\;p^{(t)}_{k,2}\,]^{\mathrm{T}}\in\mathbb{R}^{2}$, where $p^{(t)}_{k,2}>0$. Each UE moves at a nominal speed $v_{\mathrm{UE}}\geq 0$, with a heading angle $\phi^{(t)}_{k}$ that evolves as a random walk:
\begin{align}
  \phi^{(t)}_{k} &=\phi^{(t-1)}_{k}+\nu^{(t)}_{k},
  \label{eq:ue_heading}\\[2pt]
  \mathbf{p}^{(t)}_{k} &=\mathbf{p}^{(t-1)}_{k}
    +v_{\mathrm{UE}}T_s\big[\cos\phi^{(t)}_{k}\;\;\sin\phi^{(t)}_{k}\big]^{\mathrm{T}},
  \label{eq:ue_mobility}
\end{align}
where $\nu^{(t)}_{k}\sim\mathcal{N}(0,\sigma_{\phi}^{2})$ and $\sigma_{\phi}$ controls the smoothness of the UE trajectory. The initial headings are independently drawn from $\mathcal{U}[0,2\pi)$. If a proposed step leaves the region defined by $20\le\|\mathbf{p}_k^{(t)}\|_2\le200$~m and a polar angle between $15^\circ$ and $165^\circ$, the UE remains at its previous position for that slot and reverses its heading by $\pi$.

\noindent\textbf{LoS Path.}
The distance and LoS AoD of UE $k$ are determined by its position as
\begin{equation}
  r^{(t)}_{k}=\big\|\mathbf{p}^{(t)}_{k}\big\|_{2},\qquad
  \cos\theta^{\mathrm{LoS},(t)}_{k}={p^{(t)}_{k,1}}/{r^{(t)}_{k}},
  \label{eq:los_geometry}
\end{equation}
and the corresponding LoS gain is
\begin{equation}
  \alpha^{\mathrm{LoS},(t)}_{k}
  =\sqrt{\beta^{(t)}_{k}\rho^{\mathrm{LoS}}_{k}}\;
   e^{-\jmath\frac{2\pi}{\lambda}r^{(t)}_{k}},
  \label{eq:los_evolution}
\end{equation}
where $\beta^{(t)}_{k}$ denotes the large-scale path gain at distance $r^{(t)}_{k}$. The LoS phase evolves according to the propagation distance, with its per-slot phase rotation $2\pi|r^{(t)}_{k}-r^{(t-1)}_{k}|/\lambda$ bounded by $2\pi f_{D}T_s$, where $f_{D}\triangleq v_{\mathrm{UE}}/\lambda$ denotes the maximum Doppler frequency. For $v_{\mathrm{UE}}T_s\ll r^{(t)}_{k}$, the per-slot change in the LoS AoD is on the order of $v_{\mathrm{UE}}T_s/r^{(t)}_{k}$.

\noindent\textbf{NLoS Paths.}
Each NLoS component represents an effective scattering cluster that is static relative to the BS, so its AoD $\theta^{\mathrm{NLoS}}_{k,\ell}$ is treated as constant over the considered horizon. The motion of the UE relative to the unresolved scatterers within each cluster induces small-scale fading, which we approximate by a first-order autoregressive (AR(1)) process with a Jakes-type one-slot correlation:
\begin{equation}
  \alpha^{\mathrm{NLoS},(t)}_{k,\ell}
  =\rho_{\alpha}\,\alpha^{\mathrm{NLoS},(t-1)}_{k,\ell}
    +\sqrt{1-\rho_{\alpha}^{2}}\;\varpi^{(t)}_{k,\ell},
  \label{eq:gain_evolution}
\end{equation}
where $\varpi^{(t)}_{k,\ell}\sim\mathcal{CN}(0,\beta^{(t)}_{k}\rho^{\mathrm{NLoS}}_{k})$ and $\rho_{\alpha}=J_{0}(2\pi f_{D}T_s)$, with $J_{0}(\cdot)$ denoting the zeroth-order Bessel function of the first kind. At initialization and following each user-replacement event, the NLoS path gains are independently drawn from $\mathcal{CN}(0,\beta_k^{(t)}\rho_k^{\mathrm{NLoS}})$. Since $\beta^{(t)}_{k}$ varies slowly relative to the small-scale fading, the gain process is approximately locally stationary with power $\beta^{(t)}_{k}\rho^{\mathrm{NLoS}}_{k}$.

For mobile UEs, the channel therefore exhibits two characteristic timescales governed by the UE velocity. The LoS AoD evolves on a geometric timescale of order $r^{(t)}_{k}/v_{\mathrm{UE}}$, whereas the LoS phase and NLoS gains vary on a Doppler timescale of order $1/f_{D}=\lambda/v_{\mathrm{UE}}$. We refer to this mobility-driven evolution as \emph{continuous channel variation}.

In addition, user arrivals, departures, or scheduling changes can alter the set of served UEs. We model such \emph{event-driven variation} by redrawing the field-response parameters of an affected UE independently, rather than evolving them from their previous values. The events follow a renewal process whose rate is specified in Section~\ref{sec:simulation}. The two variation models represent, respectively, continuously changing channels of slowly moving UEs and occasional changes in the set of stationary served UEs; they are evaluated separately in Section~\ref{sec:simulation}. We denote the field-response parameters at slot $t$ by $\bm{\xi}^{(t)}$ and write $\mathbf{H}_{t}(\mathbf{x})\triangleq\mathbf{H}(\mathbf{x};\bm{\xi}^{(t)})$ for brevity.

\subsection{CSI Acquisition}\label{subsec:csi}

\noindent\textbf{Ideal-CSI model.}
The algorithm and theoretical analysis assume that the current field-response parameters $\bm{\xi}^{(t)}$ are available when selecting the slot-$t$ APV, enabling evaluation of candidate positions via~\eqref{eq:channel}. Existing channel acquisition methods~\cite{xiao2024channel,xu2024fas,zheng2026trajectory} are complementary to our energy-allocation framework. We also consider a pilot-based operating mode without dedicated mechanical CSI-probing movements.

\noindent\textbf{Pilot-based operation in sparse channels.}
At the settled APV $\mathbf{x}_t$, $K$ orthogonal TDD uplink pilots yield the matched-filter channel observation
\begin{equation}
\begin{aligned}
 \mathbf{Y}_t&=\mathbf{H}_t(\mathbf{x}_t)+\mathbf{E}_t,\\
 [\mathbf{E}_t]_{mk}&\sim\mathcal{CN}(0,\sigma_{e,k}^{2}),\qquad
 \sigma_{e,k}^{2}=\frac{\sigma_n^2}{P_{\mathrm{pil},k}N_{\mathrm{pil}}},
\end{aligned}
\label{eq:pilot_observation}
\end{equation}
where $N_{\mathrm{pil}}\ge K$ and $P_{\mathrm{pil},k}$ denote the orthogonal pilot length and UE pilot power, respectively. The direct least-squares (LS) estimate $\hat{\mathbf H}_t^{\rm LS}=\mathbf Y_t$ supports ZF precoding at $\mathbf{x}_t$, while orthogonal matching pursuit (OMP) extracts a sparse local field-response model from the same observation to determine $\mathbf{x}_{t+1}$. The OMP computation runs in parallel with data transmission. Since CSI acquisition requires no dedicated mechanical antenna sweeps, it incurs no additional actuation-energy consumption. A newly scheduled user can therefore influence repositioning from the slot following its first pilot. Section~\ref{subsec:estimated_csi} evaluates this operating mode.

\section{Problem Formulation}
\label{sec:formulation}

We first characterize the set of APVs reachable within a single slot, then define the per-slot utility and movement energy, and finally state a long-term throughput maximization problem subject to a movement energy budget.

\subsection{Per-Slot Feasible Set}
\label{subsec:feasible}

Given the previous APV $\mathbf{x}_{t-1}$, the APV of slot $t$ must lie in the feasible set
\begin{align}
  &\mathcal{A}(\mathbf{x}_{t-1})\triangleq\nonumber\\
  &\left\{\mathbf{x}\in\mathbb{R}^{M}\;\middle|\;
  \begin{aligned}
    &0\le x_{1},\quad x_{M}\le D,\\
    &x_{m+1}-x_{m}\ge d,\ \forall m\in\{1,\dots,M-1\},\\
    &\big\|\mathbf{x}-\mathbf{x}_{t-1}\big\|_{\infty}\le vT_d
  \end{aligned}\right\},
  \label{eq:feasible_set}
\end{align}
where $d>0$ denotes the minimum inter-element spacing imposed to mitigate mutual coupling, and the initial APV $\mathbf{x}_{0}$ is assumed to satisfy the aperture and spacing constraints. The spacing constraint also enforces the ascending order $x_{1}<x_{2}<\cdots<x_{M}$, so that the antenna indexing remains consistent across slots. The last condition is equivalent to $\tau^{\mathrm{mov}}_{t}\le T_d$, so that repositioning and any pilot interval fit within the slot.

\begin{remark}[Properties of $\mathcal{A}(\mathbf{x}_{t-1})$]
\label{rem:compact}
The set $\mathcal{A}(\mathbf{x}_{t-1})$ is the intersection of a box, $M-1$ linear half-spaces, and an $\ell_{\infty}$ ball centered at $\mathbf{x}_{t-1}$, and is therefore compact and convex. Moreover, it is nonempty for every $t\ge1$: since $\mathbf{x}_{0}$ satisfies the first two conditions in \eqref{eq:feasible_set}, so does every $\mathbf{x}_{t-1}$ by induction, and hence the \emph{stay-put} action $\mathbf{x}_{t}=\mathbf{x}_{t-1}$ is always feasible.
\end{remark}

\subsection{Long-Term Throughput Maximization}
\label{subsec:formulation}

Let $\mathbf{W}(\mathbf{x};\bm{\xi}^{(t)})=[\,\mathbf{w}_{1}\;\cdots\;\mathbf{w}_{K}\,]$ denote a ZF precoder\footnote{The algorithm of Section~\ref{sec:algorithm} also applies to regularized ZF (RZF) and maximum-ratio transmission (MRT) precoding, provided $\tilde{\gamma}^{\mathrm{sum}}$ is twice differentiable in $\mathbf{x}$.} with unit-norm columns, given by the normalized columns of $\mathbf{H}(\mathbf{H}^{\mathrm{H}}\mathbf{H})^{-1}$. With equal power allocation under the total power $P$, the instantaneous sum-rate is
\begin{equation}
  \tilde{\gamma}^{\mathrm{sum}}(\mathbf{x};\bm{\xi}^{(t)})
  =\sum_{k=1}^{K}\log_{2}\!\left(1+
  \frac{\frac{P}{K}\,|\mathbf{h}_{k}^{\mathrm{H}}\mathbf{w}_{k}|^{2}}
  {\frac{P}{K}\sum_{j\ne k}|\mathbf{h}_{k}^{\mathrm{H}}\mathbf{w}_{j}|^{2}+\sigma_{n}^{2}}
  \right),
  \label{eq:sumrate}
\end{equation}
where the dependence of $\mathbf{h}_{k}$ and $\mathbf{w}_{k}$ on $(\mathbf{x};\bm{\xi}^{(t)})$ is omitted for brevity and $\sigma_{n}^{2}$ denotes the noise power. The precoder is recomputed from the position-dependent channel matrix at every candidate APV, thereby accounting for the coupling between antenna placement and beamforming. We do not treat the beamforming vectors as additional optimization variables, since our focus is causal antenna repositioning under inter-slot reachability and a long-term movement-energy budget; joint beamforming optimization would introduce an additional nonconvex subproblem in each slot. Since the ZF precoders have unit-norm columns, $|\mathbf{h}_k^{\mathrm H}\mathbf{w}_k|^2\le\|\mathbf{h}_k\|_2^2\le Mb_{k,t}^2$, where $b_{k,t}$ is the sum of the magnitudes of UE $k$'s path gains. Thus, a realization-dependent upper bound on the sum-rate, valid at every APV where the precoder is defined, is
\begin{align}
  b_{k,t}&\triangleq|\alpha_k^{\mathrm{LoS},(t)}|+\sum_{\ell=1}^{L_k}|\alpha_{k,\ell}^{\mathrm{NLoS},(t)}|,\nonumber\\
  \hat\Gamma_t&\triangleq\sum_{k=1}^{K}\log_2\!\left(1+\frac{PMb_{k,t}^2}{K\sigma_n^2}\right)<\infty.
  \label{eq:realized_rate_bound}
\end{align}
Unlike a uniform bound over all channel realizations, $\hat\Gamma_t$ is finite for each realized set of path gains, including those generated by the Gaussian model in \eqref{eq:gain_evolution}. The per-slot utility, defined as the effective throughput of slot $t$, is given by
\begin{equation}
U_{t}(\mathbf{x}_{t},\mathbf{x}_{t-1})
\triangleq \frac{T_d-\|\mathbf{x}_{t}-\mathbf{x}_{t-1}\|_{\infty}/v}{T_s}
\tilde{\gamma}^{\mathrm{sum}}(\mathbf{x}_{t};\bm{\xi}^{(t)}).
\label{eq:utility}
\end{equation}
The prefactor in \eqref{eq:utility} is the data-transmission time normalized by the \emph{full} slot duration $T_s$, not by $T_d$. Thus, $U_t$ is an effective spectral efficiency measured in bit/s/Hz; changing $\tau_{\mathrm p}$ does not alter the normalization.
In contrast to the movement delay in \eqref{eq:mov_delay}, which is governed by the \emph{maximum} displacement, the mechanical energy for repositioning is consumed by \emph{every} actuator independently. Accordingly, the movement energy of slot $t$ is modeled as
\begin{equation}
  e_{t}(\mathbf{x}_{t},\mathbf{x}_{t-1})
  \triangleq
  c_{\mathrm{e}}\big\|\mathbf{x}_{t}-\mathbf{x}_{t-1}\big\|_{1},
  \label{eq:energy}
\end{equation}
where $c_{\mathrm{e}}>0$ denotes the energy consumption per unit displacement.

We aim to maximize the time-averaged effective throughput subject to a time-averaged movement energy budget:
\begin{subequations}\label{eq:P1}
\begin{align}
  \mathcal{P}_{1}:
  \max_{\{\mathbf{x}_{t}\}_{t\ge1}}\quad
  & \bar{U}\;\triangleq\;
    \liminf_{T\to\infty}\frac{1}{T}\sum_{t=1}^{T}
    \mathbb{E}\big[U_{t}(\mathbf{x}_{t},\mathbf{x}_{t-1})\big]
    \label{eq:P1_obj}\\
  \mathrm{s.t.}\quad
  & \bar{e}\;\triangleq\;
    \limsup_{T\to\infty}\frac{1}{T}\sum_{t=1}^{T}
    \mathbb{E}\big[e_{t}(\mathbf{x}_{t},\mathbf{x}_{t-1})\big]
    \;\le\; E_{\mathrm{avg}},
    \label{eq:P1_energy}\\
  & \mathbf{x}_{t}\in\mathcal{A}(\mathbf{x}_{t-1}),
    \quad\forall t\ge1,
    \label{eq:P1_feasible}
\end{align}
\end{subequations}
where $E_{\mathrm{avg}}>0$ denotes the average movement energy budget, the optimization is over causal policies in which $\mathbf{x}_{t}$ depends only on the information available up to slot $t$, and the expectation is taken over the channel process and any randomness in the policy.

Since movement delay depends on the maximum antenna displacement whereas actuation energy accumulates over all displacements, moving additional antennas within the same displacement range incurs no further delay but consumes additional energy. This asymmetry creates an incentive for sparse repositioning under the energy constraint, as exploited by the proximal update in Section~\ref{subsec:prox}.

Problem $\mathcal{P}_{1}$ is challenging for three reasons. First, the decisions are coupled across slots, both through the previous APV $\mathbf{x}_{t-1}$ in \eqref{eq:feasible_set}, \eqref{eq:utility}, and \eqref{eq:energy}, and through the time-averaged constraint \eqref{eq:P1_energy}. Second, the achievable rate is non-concave in $\mathbf{x}_{t}$, while movement delay and energy introduce the non-smooth $\ell_{\infty}$ and $\ell_{1}$ norms, respectively. Third, an optimal causal policy would require the statistics of the channel process and the solution of a constrained Markov decision problem over a continuous state space, which is intractable in practice. To address these challenges, we adopt the Lyapunov optimization framework~\cite{neely2010}, which handles the time-averaged constraint through a virtual queue and requires only the current channel state and the previous APV, without any statistical knowledge of the channel process.

\section{Lyapunov-Based Online Framework}
\label{sec:framework}

This section transforms $\mathcal{P}_{1}$ into a sequence of per-slot problems based on the Lyapunov drift-plus-penalty method~\cite{neely2010}, in which a virtual queue prices the movement energy against the instantaneous throughput.

To handle the time-averaged constraint \eqref{eq:P1_energy}, we introduce a \emph{virtual queue} $Q_{t}$ with $Q_{1}=0$, which evolves as
\begin{equation}
  Q_{t+1}=\Big[\,Q_{t}+e_{t}(\mathbf{x}_{t},\mathbf{x}_{t-1})
  -E_{\mathrm{avg}}\,\Big]^{+}.
  \label{eq:virtual_queue}
\end{equation}
The backlog $Q_{t}$ measures the accumulated excess of the movement energy over the budget, and \eqref{eq:P1_energy} is satisfied if $Q_{t}$ grows sublinearly in $t$. The reachability constraint further yields the following bounds on the per-slot quantities.

\begin{lemma}[Per-slot bounds]\label{lem:bounded}
For all $t\ge1$ and all $\mathbf{x}_{t}\in\mathcal{A}(\mathbf{x}_{t-1})$,
\begin{align}
  0& \le U_{t} \le \hat\Gamma_t,\nonumber\\
  0&\le e_{t}\le e_{\max}\triangleq c_{\mathrm{e}}MvT_d.
  \label{eq:bounded}
\end{align}
\end{lemma}
\begin{IEEEproof}
The reachability constraint in \eqref{eq:feasible_set} gives
$\|\mathbf{x}_{t}-\mathbf{x}_{t-1}\|_{\infty}\le vT_d$, so that the fraction
of the slot available for data transmission satisfies
$(T_d-\|\mathbf{x}_{t}-\mathbf{x}_{t-1}\|_{\infty}/v)/T_s\in[0,1]$. Since
$0\le\tilde{\gamma}^{\mathrm{sum}}(\mathbf{x};\bm{\xi}^{(t)})\le\hat\Gamma_t$ by \eqref{eq:realized_rate_bound}, the product in \eqref{eq:utility} satisfies $0\le U_{t}\le\hat\Gamma_t$. For the energy, $\|\mathbf{u}\|_{1}\le M\|\mathbf{u}\|_{\infty}$ and the reachability constraint yield $e_{t}=c_{\mathrm{e}}\|\mathbf{x}_{t}-\mathbf{x}_{t-1}\|_{1} \le c_{\mathrm{e}}M\|\mathbf{x}_{t}-\mathbf{x}_{t-1}\|_{\infty} \le c_{\mathrm{e}}MvT_d$.
\end{IEEEproof}

With the Lyapunov function $\mathcal{L}(Q_{t})\triangleq\frac{1}{2}Q_{t}^{2}$ and the one-slot conditional drift $\Delta(Q_{t})\triangleq \mathbb{E}\big[\mathcal{L}(Q_{t+1})-\mathcal{L}(Q_{t})\,\big|\,Q_{t}\big]$, the drift-plus-penalty (DPP) term admits the following bound, in which $V>0$ is a control parameter balancing throughput optimality against queue backlog.

\begin{lemma}[Drift-plus-penalty bound]\label{lem:drift}
For any $V>0$ and any feasible sequence $\{\mathbf{x}_{t}\}$,
\begin{equation}
  \Delta(Q_{t})-V\,\mathbb{E}\big[U_{t}\,\big|\,Q_{t}\big]
  \;\le\;
  C_{0} +\mathbb{E}\Big[Q_{t}\big(e_{t}-E_{\mathrm{avg}}\big)-V\,U_{t}\;\Big|\;Q_{t}\Big],
  \label{eq:dpp}
\end{equation}
where $C_{0}\triangleq\frac{1}{2}\big(e_{\max}^{2}+E_{\mathrm{avg}}^{2}\big)$.
\end{lemma}
\begin{IEEEproof}
Using $([\,a\,]^{+})^{2}\le a^{2}$ for any $a\in\mathbb{R}$, the queue update \eqref{eq:virtual_queue} gives
\begin{align}
  \tfrac{1}{2}Q_{t+1}^{2}-\tfrac{1}{2}Q_{t}^{2}
  &\le\tfrac{1}{2}\big(e_{t}-E_{\mathrm{avg}}\big)^{2}
    +Q_{t}\big(e_{t}-E_{\mathrm{avg}}\big)\notag\\
  &\le C_{0}+Q_{t}\big(e_{t}-E_{\mathrm{avg}}\big),
  \label{eq:drift_raw}
\end{align}
where the last inequality follows from $(e_{t}-E_{\mathrm{avg}})^{2}\le e_{t}^{2}+E_{\mathrm{avg}}^{2}\le e_{\max}^{2}+E_{\mathrm{avg}}^{2}=2C_{0}$, using $e_{t},E_{\mathrm{avg}}\ge0$ and Lemma~\ref{lem:bounded}. Taking the conditional expectation given $Q_{t}$ and subtracting $V\,\mathbb{E}[U_{t}\,|\,Q_{t}]$ from both sides yields \eqref{eq:dpp}.
\end{IEEEproof}

Following the DPP principle, the APV at each slot is chosen to minimize the right-hand side of \eqref{eq:dpp}. Discarding the terms independent of $\mathbf{x}_{t}$ and negating the objective yields the per-slot problem
\begin{equation}
  \mathcal{P}_{2}:\;\; \max_{\mathbf{x}\in\mathcal{A}(\mathbf{x}_{t-1})}\; F(\mathbf{x})\triangleq h(\mathbf{x})-\psi(\mathbf{x}),
  \label{eq:P2}
\end{equation}
where
\begin{align}
  h(\mathbf{x}) &\triangleq \frac{V}{T_s}\Big(T_d-\frac{\|\mathbf{x}-\mathbf{x}_{t-1}\|_{\infty}}{v}\Big)
\tilde{\gamma}^{\mathrm{sum}}(\mathbf{x};\bm{\xi}^{(t)}),
  \label{eq:h_def}\\
  \psi(\mathbf{x})
  &\triangleq Q_{t}c_{\mathrm{e}}\|\mathbf{x}-\mathbf{x}_{t-1}\|_{1}
  \label{eq:psi_def}
\end{align} denote the weighted effective throughput and the energy penalty, respectively. 

Problem $\mathcal{P}_{2}$ requires only the current field response, previous APV, and queue; it does not require channel statistics. The queue $Q_t$ adaptively prices energy: excess past consumption raises the price and suppresses movement, whereas unused budget reduces the price and favors throughput. Thus repositioning is governed by the \emph{history} of energy use.

Unlike conventional drift-plus-penalty models with state-independent feasible actions~\cite{neely2010}, LOMA's reachable set depends on its previous APV. Accordingly, our performance analysis separates a sample-path energy bound (Theorem~\ref{thm:energy}) from a \emph{state-matched} throughput bound (Proposition~\ref{prop:reallocation}), while independently evolving policies are compared numerically.

\section{LOMA Algorithm}
\label{sec:algorithm}

LOMA approximately solves $\mathcal P_2$ through log-sum-exp smoothing of the nonsmooth movement delay, a minorization--maximization (MM) surrogate for the nonconcave effective-throughput term, and a proximal update retaining the queue-weighted $\ell_1$ movement penalty. The smoothing and MM construction build on the single-block delay-aware design in~\cite{tae2026vtc}, while the proposed proximal update jointly handles the queue-weighted movement penalty and inter-slot reachability constraints. This structure promotes sparse antenna repositioning (Remark~\ref{rem:sparse}). We abbreviate $\mathbf{x}^{\mathrm{prev}}=\mathbf{x}_{t-1}$ and $\mathcal A=\mathcal A(\mathbf{x}_{t-1})$ below.

\subsection{LSE Smoothing of the Movement Delay Term}
\label{subsec:lse}

Let $\mathbf{u}\triangleq\mathbf{x}-\mathbf{x}^{\mathrm{prev}}$. Following the smoothing technique in~\cite{nesterov2005}, the $\ell_{\infty}$ norm in \eqref{eq:h_def} is approximated by smoothing each absolute value and then applying the LSE function, i.e.,
\begin{equation}
  \Phi_{\mu,\epsilon}(\mathbf{u})
  \triangleq
  \mu\ln\!\left(\sum_{m=1}^{M}
  \exp\!\left(\frac{\sqrt{u_{m}^{2}+\epsilon^{2}}}{\mu}\right)\right),
  \label{eq:lse}
\end{equation}
where $\mu>0$ and $\epsilon>0$ denote the LSE and absolute-value smoothing parameters, respectively. The resulting smoothed counterpart of $h(\mathbf{x})$ is
\begin{equation}
 \tilde{h}(\mathbf{x}) \triangleq \frac{V}{T_s}\Big(T_d-\frac{\Phi_{\mu,\epsilon}(\mathbf{x}-\mathbf{x}^{\mathrm{prev}})}{v}\Big)
\tilde{\gamma}^{\mathrm{sum}}(\mathbf{x};\bm{\xi}^{(t)}),
  \label{eq:h_tilde}
\end{equation}
and the smoothed objective is defined as $\tilde{F}(\mathbf{x})\triangleq\tilde{h}(\mathbf{x})-\psi(\mathbf{x})$. The following proposition shows that $\tilde{F}$ is a lower bound on $F$ with a controllable gap.

\begin{proposition}[Conservative smoothing]\label{prop:lse_bound}
For all $\mathbf{x}\in\mathcal{A}$, the smoothing gap satisfies
\begin{equation}
  0 \le F(\mathbf{x})-\tilde{F}(\mathbf{x}) \le
\frac{V\hat\Gamma_t}{vT_s}\big(\epsilon+\mu\ln M\big).
  \label{eq:lse_gap}
\end{equation}
\end{proposition}
\begin{IEEEproof}
For any $\mathbf{a}\in\mathbb{R}^{M}$, the LSE function satisfies~\cite{nesterov2005} $\max_{m}a_{m}\le\mu\ln\sum_{m}e^{a_{m}/\mu}\le\max_{m}a_{m}+\mu\ln M$. Setting $a_{m}=\sqrt{u_{m}^{2}+\epsilon^{2}}$ and using $|u_{m}|\le\sqrt{u_{m}^{2}+\epsilon^{2}}\le|u_{m}|+\epsilon$ together with the monotonicity of the maximum yields
\begin{equation}
  \|\mathbf{u}\|_{\infty}
  \;\le\;\Phi_{\mu,\epsilon}(\mathbf{u})
  \;\le\;\|\mathbf{u}\|_{\infty}+\epsilon+\mu\ln M .
  \label{eq:lse_bound}
\end{equation}
Since $\psi(\mathbf{x})$ is common to $F$ and $\tilde{F}$, it follows from \eqref{eq:h_def} and \eqref{eq:h_tilde} that
\begin{equation}\label{eq:diff}
F(\mathbf{x})-\tilde{F}(\mathbf{x}) = \frac{V}{vT_s}\,
\tilde{\gamma}^{\mathrm{sum}}(\mathbf{x};\bm{\xi}^{(t)})
\Big(\Phi_{\mu,\epsilon}(\mathbf{u})-\|\mathbf{u}\|_{\infty}\Big).
\end{equation}
Combining \eqref{eq:lse_bound} and \eqref{eq:diff} with $0\le\tilde{\gamma}^{\mathrm{sum}}(\mathbf{x};\bm{\xi}^{(t)})\le\hat\Gamma_t$ establishes \eqref{eq:lse_gap}.
\end{IEEEproof}

Proposition~\ref{prop:lse_bound} shows that the smoothing gap can be made arbitrarily small by reducing $\epsilon$ and $\mu$, at the cost of increased curvature and hence a larger Lipschitz constant in Section~\ref{subsec:mm}. We balance the two approximation terms by setting $\mu=\epsilon/\ln M$ and choose $\epsilon$ as a small fraction of the ideal-slot reachable range $vT_s$. To ensure that the smoothed objective remains nonnegative at the stay-put action, we impose $\epsilon+\mu\ln M\le vT_d$, including when $\tau_{\mathrm p}>0$. This condition is used in the energy guarantee of Theorem~\ref{thm:energy}.

\subsection{Minorization--Maximization Surrogate}
\label{subsec:mm}

Since $\tilde{h}(\mathbf{x})$ remains non-concave, we construct a surrogate function that minorizes $\tilde{F}(\mathbf{x})$. Let $\mathbf{x}^{(i-1)}$ denote the iterate at the $(i\!-\!1)$-th inner iteration, with $\mathbf{x}^{(0)}=\mathbf{x}^{\mathrm{prev}}$.

Throughout, $\tilde{\gamma}^{\mathrm{sum}}(\cdot;\bm{\xi}^{(t)})$ is assumed to be twice continuously differentiable on a neighborhood of $\mathcal{A}$. This holds for MRT precoding whenever $\mathbf{h}_{k}(\mathbf{x};\bm{\xi}^{(t)})\neq \mathbf{0}$ for all $k$, and for ZF precoding whenever $\mathbf{H}_{t}(\mathbf{x})$ has full column rank on $\mathcal{A}$. Since $\Phi_{\mu,\epsilon}$ is infinitely differentiable for any $\mu,\epsilon>0$, $\tilde{h}$ is then twice continuously differentiable, and the compactness of $\mathcal{A}$ implies that $\nabla\tilde{h}$ is Lipschitz continuous on $\mathcal{A}$ with some finite constant $L_{\mathrm{loc}}$. For the exact-gradient algorithm, this full-column-rank condition supports the finite backtracking-trial bound and stationary-limit-point guarantee in Proposition~\ref{prop:convergence}; the monotonicity of accepted iterates and the energy bound in Theorem~\ref{thm:energy} do not require it, provided that the per-slot objective is well defined and the algorithm returns an accepted iterate or the stay-put APV. Moreover, since $\mathcal{A}$ is convex, the descent lemma~\cite[Prop.~A.24]{bertsekas1999} yields, for any $L\ge L_{\mathrm{loc}}$,
\begin{align}
  \tilde{h}(\mathbf{x})\;\ge\;
  &\tilde{h}(\mathbf{x}^{(i-1)})
  +\nabla\tilde{h}(\mathbf{x}^{(i-1)})^{\mathrm{T}}
   \big(\mathbf{x}-\mathbf{x}^{(i-1)}\big)\nonumber\\
  &-\tfrac{L}{2}\big\|\mathbf{x}-\mathbf{x}^{(i-1)}\big\|_{2}^{2},
  \quad\forall\mathbf{x}\in\mathcal{A},
  \label{eq:minorant_h}
\end{align}
with equality at $\mathbf{x}=\mathbf{x}^{(i-1)}$. Accordingly, the surrogate function is defined as
\begin{align}
  g\big(\mathbf{x}\,\big|\,\mathbf{x}^{(i-1)};L\big)
  &\triangleq
  \tilde{h}(\mathbf{x}^{(i-1)})
  +\nabla\tilde{h}(\mathbf{x}^{(i-1)})^{\mathrm{T}}
   \big(\mathbf{x}-\mathbf{x}^{(i-1)}\big)\nonumber\\
  &\quad-\tfrac{L}{2}\big\|\mathbf{x}-\mathbf{x}^{(i-1)}\big\|_{2}^{2}
  -\psi(\mathbf{x}),
  \label{eq:surrogate}
\end{align}
where the concave term $-\psi(\mathbf{x})$ is retained without approximation. For $L\ge L_{\mathrm{loc}}$, \eqref{eq:minorant_h} implies that the surrogate \eqref{eq:surrogate} satisfies the minorization conditions
\begin{align}
  g\big(\mathbf{x}\,\big|\,\mathbf{x}^{(i-1)};L\big)
  &\;\le\;\tilde{F}(\mathbf{x}),
  \qquad\forall\mathbf{x}\in\mathcal{A},\nonumber\\
  g\big(\mathbf{x}^{(i-1)}\,\big|\,\mathbf{x}^{(i-1)};L\big)
  &\;=\;\tilde{F}(\mathbf{x}^{(i-1)}).
  \label{eq:mm_conditions}
\end{align}
Furthermore, since $\psi$ is convex, \eqref{eq:surrogate} is strongly concave in $\mathbf{x}$ with modulus $L$, and its maximizer over $\mathcal{A}$ is unique.

\subsection{Surrogate Maximization via Proximal Update}
\label{subsec:prox}

Discarding the terms independent of $\mathbf{x}$ in \eqref{eq:surrogate} and completing the square, maximizing the surrogate over $\mathcal{A}$ is equivalent to the constrained proximal problem
\begin{equation}
  \min_{\mathbf{x}\in\mathcal{A}}\;
  \tfrac{L}{2}\big\|\mathbf{x}-\mathbf{z}^{(i)}\big\|_{2}^{2}
  +\psi(\mathbf{x}),
  \label{eq:prox_problem}
\end{equation}
where $\mathbf{z}^{(i)}\triangleq\mathbf{x}^{(i-1)}+\tfrac{1}{L}\nabla\tilde{h}(\mathbf{x}^{(i-1)})$ denotes the tentative gradient step.

\noindent\textbf{Exact Proximal Update.}
Introducing an auxiliary vector $\mathbf{s}\in\mathbb{R}^{M}$ that upper-bounds the elementwise displacements, \eqref{eq:prox_problem} can be equivalently formulated as the quadratic program (QP)
\begin{subequations}\label{eq:prox_qp}
\begin{align}
  \min_{\mathbf{x},\,\mathbf{s}}\;\;
  &\tfrac{L}{2}\big\|\mathbf{x}-\mathbf{z}^{(i)}\big\|_{2}^{2}
   +Q_{t}c_{\mathrm{e}}\mathbf{1}^{\mathrm{T}}\mathbf{s}
   \label{eq:prox_qp_obj}\\
  \mathrm{s.t.}\;\;
  &x_{m}^{\min}\le x_{m}\le x_{m}^{\max},\quad\forall m,
    \label{eq:prox_qp_box}\\
  &x_{m+1}-x_{m}\ge d,\quad\forall m\in\{1,\dots,M-1\},
    \label{eq:prox_qp_space}\\
  &-\mathbf{s}\le\mathbf{x}-\mathbf{x}^{\mathrm{prev}}\le\mathbf{s},
    \label{eq:prox_qp_abs}
\end{align}
\end{subequations}
where the per-element bounds induced by \eqref{eq:feasible_set} are
\begin{align}
  x_{m}^{\min}=\max\big(0,\,x^{\mathrm{prev}}_{m}-vT_d\big),\nonumber\\
  x_{m}^{\max}=\min\big(D,\,x^{\mathrm{prev}}_{m}+vT_d\big).
  \label{eq:proj_bounds}
\end{align}
Problem \eqref{eq:prox_qp} has $2M$ variables and $5M-1$ linear constraints. It is feasible because the stay-put action $\mathbf{x}=\mathbf{x}^{\mathrm{prev}}$, $\mathbf{s}=\mathbf{0}$ satisfies all constraints. Since \eqref{eq:prox_problem} is strongly convex, the $\mathbf{x}$-component of the QP solution, denoted by $\mathbf{x}^{\mathrm{tmp}}$, is unique and coincides with the solution of \eqref{eq:prox_problem}.

\noindent\textbf{Backtracking Line Search.}
The Lipschitz constant $L_{\mathrm{loc}}$ in \eqref{eq:minorant_h} is difficult to determine analytically. We therefore employ backtracking that increases $L\leftarrow\eta L$ with $\eta>1$ until the acceptance condition
\begin{equation}
  \tilde{F}\big(\mathbf{x}^{\mathrm{tmp}}\big)
  \;\ge\;
  g\big(\mathbf{x}^{\mathrm{tmp}}\,\big|\,\mathbf{x}^{(i-1)};L\big)
  \label{eq:acceptance_cond}
\end{equation}
is satisfied, where $\mathbf{x}^{\mathrm{tmp}}$ solves \eqref{eq:prox_qp} for the current $L$. We impose $L\ge L_{\mathrm{init}}V$ with $L_{\mathrm{init}}>0$, consistent with the linear scaling of $\nabla\tilde{h}$ in $V$. To reduce the number of trials, the search is warm-started from the previously accepted value divided by $\eta$, subject to this lower bound. For the exact-gradient algorithm, \eqref{eq:acceptance_cond} is guaranteed whenever $L\ge L_{\mathrm{loc}}$. Hence, at most $\lceil\log_{\eta}(L_{\mathrm{loc}}/(L_{\mathrm{init}}V))\rceil^{+}+1$ trials are required per iteration, and every accepted value within slot $t$ satisfies $L_{\mathrm{init}}V\le L\le\bar{L}\triangleq\max\{\eta L_{\mathrm{loc}},L_{\mathrm{start}}\}$, where $L_{\mathrm{start}}$ denotes the value of $L$ at the beginning of the slot.

Furthermore, each accepted iterate satisfies
\begin{align}
  \tilde{F}\big(\mathbf{x}^{(i)}\big)
  &\;\ge\;g\big(\mathbf{x}^{(i)}\,\big|\,\mathbf{x}^{(i-1)};L\big)\nonumber\\
  &\;\ge\;g\big(\mathbf{x}^{(i-1)}\,\big|\,\mathbf{x}^{(i-1)};L\big)
   +\tfrac{L}{2}\big\|\mathbf{x}^{(i)}-\mathbf{x}^{(i-1)}\big\|_{2}^{2}\nonumber\\
  &\;=\;\tilde{F}\big(\mathbf{x}^{(i-1)}\big)
   +\tfrac{L}{2}\big\|\mathbf{x}^{(i)}-\mathbf{x}^{(i-1)}\big\|_{2}^{2},
  \label{eq:monotone}
\end{align}
where the first inequality follows from \eqref{eq:acceptance_cond}, the second from the strong concavity of the surrogate, and the equality from \eqref{eq:mm_conditions}. Thus, the smoothed objective increases monotonically. We impose a maximum of $J_{\max}$ backtracking trials; if this limit is reached, the algorithm retains $\mathbf{x}^{(i-1)}$, preserving \eqref{eq:monotone} with equality.

\begin{remark}[Sparse repositioning]\label{rem:sparse}
The proximal update in \eqref{eq:prox_qp} exhibits a soft-thresholding structure induced by the queue-weighted $\ell_1$ movement penalty. Let $\tau\triangleq Q_{t}c_{\mathrm{e}}/L$ and $\mathcal{S}_{\tau}(u)\triangleq\mathrm{sign}(u)\,[\,|u|-\tau\,]^{+}$. If no spacing constraint \eqref{eq:prox_qp_space} involving antenna $m$ is active at $\mathbf{x}^{\mathrm{tmp}}$, its optimality conditions decouple, yielding
\begin{equation}
  x^{\mathrm{tmp}}_{m}
  =\Pi_{[x_{m}^{\min},\,x_{m}^{\max}]}
  \Big(x^{\mathrm{prev}}_{m}
  +\mathcal{S}_{\tau}\big(z^{(i)}_{m}-x^{\mathrm{prev}}_{m}\big)\Big),
  \label{eq:soft_threshold}
\end{equation}
where $\Pi_{[x_{m}^{\min},\,x_{m}^{\max}]}$ denotes projection onto the specified interval. Since $x_{m}^{\min}\le x^{\mathrm{prev}}_{m}\le x_{m}^{\max}$, an antenna satisfying $|z^{(i)}_{m}-x^{\mathrm{prev}}_{m}|\le\tau$ remains stationary. When $Q_{t}=0$, the update reduces to that of a purely delay-aware design. For fixed $L$ and $\mathbf{z}^{(i)}$, a larger $Q_{t}$ raises the soft-threshold level and can keep more antennas stationary, except when active spacing constraints couple their updates.
\end{remark}

\begin{algorithm}[t]
\caption{Lyapunov-Based Online Movable-Antenna Repositioning (LOMA)}
\label{alg:prox_scamm}
\begin{algorithmic}[1]
\Require Initial APV $\mathbf{x}_{0}$, control parameter $V$, energy budget $E_{\mathrm{avg}}$, tolerance $\delta$, maximum inner iterations $I_{\max}$, maximum backtracking trials $J_{\max}$, backtracking factor $\eta>1$, initial Lipschitz scale $L_{\mathrm{init}}>0$
\Ensure APV sequence $\{\mathbf{x}_{t}\}_{t\ge1}$
\State $Q_{1}\leftarrow 0$,\; $L\leftarrow L_{\mathrm{init}}V$
\For{$t=1,2,\dots$} \label{ln:outer_begin}
  \State Acquire $\bm{\xi}^{(t)}$;\; $\mathbf{x}^{\mathrm{prev}}\leftarrow\mathbf{x}_{t-1}$;\; compute~\eqref{eq:proj_bounds}
  \State $\mathbf{x}^{(0)}\leftarrow\mathbf{x}^{\mathrm{prev}}$;\;
  $\tilde{F}_{0}\leftarrow V\tilde{\gamma}^{\mathrm{sum}}(\mathbf{x}^{\mathrm{prev}};\bm{\xi}^{(t)})$
  \For{$i=1$ \textbf{to} $I_{\max}$} \label{ln:inner_begin}
    \State Evaluate $\tilde{h}(\mathbf{x}^{(i-1)})$ and
      $\nabla\tilde{h}(\mathbf{x}^{(i-1)})$ via~\eqref{eq:h_tilde}
    \State $L\leftarrow\max\{L/\eta,\,L_{\mathrm{init}}V\}$;\;
      $\mathbf{x}^{(i)}\leftarrow\mathbf{x}^{(i-1)}$
    \For{$j=1$ \textbf{to} $J_{\max}$}
      \State $\mathbf{z}^{(i)}\leftarrow\mathbf{x}^{(i-1)}
        +\tfrac{1}{L}\nabla\tilde{h}(\mathbf{x}^{(i-1)})$
      \State Obtain $\mathbf{x}^{\mathrm{tmp}}$ by solving the QP~\eqref{eq:prox_qp} \label{ln:prox}
      \If{\eqref{eq:acceptance_cond} holds}
        \State $\mathbf{x}^{(i)}\leftarrow\mathbf{x}^{\mathrm{tmp}}$;\; \textbf{break}
      \EndIf
      \State $L\leftarrow\eta L$
    \EndFor
    \State \textbf{if} $\chi^{(i)}\le\delta$ in~\eqref{eq:stopping} \textbf{then break}
  \EndFor \label{ln:inner_end}
  \State $\mathbf{x}_{t}\leftarrow\mathbf{x}^{(i)}$
  \State $Q_{t+1}\leftarrow\big[\,Q_{t}
    +c_{\mathrm{e}}\|\mathbf{x}_{t}-\mathbf{x}_{t-1}\|_{1}
    -E_{\mathrm{avg}}\,\big]^{+}$ \label{ln:queue_update}
\EndFor \label{ln:outer_end}
\end{algorithmic}
\end{algorithm}

\subsection{Overall Algorithm}
\label{subsec:overall}

The complete procedure of LOMA is summarized in Algorithm~\ref{alg:prox_scamm}. In each slot, the outer loop (lines~\ref{ln:outer_begin}--\ref{ln:outer_end}) obtains the channel information, invokes the inner loop, and updates the virtual queue. The inner loop (lines~\ref{ln:inner_begin}--\ref{ln:inner_end}) approximately solves $\mathcal{P}_{2}$ using proximal MM, starting from $\mathbf{x}^{(0)}=\mathbf{x}^{\mathrm{prev}}$, and terminates when the relative improvement
\begin{equation}
  \chi^{(i)}\triangleq
  \frac{\tilde{F}(\mathbf{x}^{(i)})-\tilde{F}(\mathbf{x}^{(i-1)})}
       {\big|\tilde{F}(\mathbf{x}^{(i-1)})\big|+\tilde{F}_{0}}
  \label{eq:stopping}
\end{equation}
falls below a tolerance $\delta>0$ or when the iteration count reaches $I_{\max}$, where $\tilde{F}_{0}\triangleq V\tilde{\gamma}^{\mathrm{sum}}(\mathbf{x}^{\mathrm{prev}};\bm{\xi}^{(t)})$ provides a reference throughput scale for normalization. The numerator is nonnegative by the monotone-ascent property in \eqref{eq:monotone}.

\begin{remark}[Energy-unconstrained operation]\label{rem:myopic}
Setting $Q_{t}\equiv0$ and skipping the queue update in line~\ref{ln:queue_update} removes the movement-energy penalty from $\mathcal{P}_{2}$. Algorithm~\ref{alg:prox_scamm} then reduces to a delay-aware, slot-by-slot effective-throughput maximization, referred to as the \emph{myopic design}. This design does not generally satisfy the long-term energy constraint \eqref{eq:P1_energy} and serves as an energy-unconstrained benchmark in Section~\ref{sec:simulation}. The monotonicity and per-slot complexity results remain valid, with $\tau=0$ in Remark~\ref{rem:sparse}.
\end{remark}

\subsection{Computational Complexity}
\label{subsec:complexity}

Let $N_{\mathrm p}\triangleq\sum_{k=1}^{K}(L_k+1)$ denote the total number of propagation paths. Under ZF precoding, each rate evaluation requires $C_{\mathrm{ZF}}\triangleq\mathcal O(MN_{\mathrm p}+MK^2+K^3)$ operations, where the $K^3$ term arises from solving the ZF linear system. In our numerical implementation, the rate gradient is approximated by central finite differences with a step size of $10^{-8}$~m, requiring $2M+1$ rate evaluations to compute the rate and its gradient. Each backtracking trial additionally evaluates the objective and solves the QP~\eqref{eq:prox_qp}. Assuming a bounded number of QP solver iterations, the dense linear-algebra cost is estimated as $\mathcal O(M^3)$ per solve~\cite{boyd2004}. Letting $I_{\mathrm{BT}}\le J_{\max}$ denote the average number of backtracking trials per inner iteration, the estimated per-slot complexity is $\mathcal O(I_{\max}[(2M+1+I_{\mathrm{BT}})C_{\mathrm{ZF}}+I_{\mathrm{BT}}M^3])$, with a worst-case bound obtained by replacing $I_{\mathrm{BT}}$ with $J_{\max}$. The virtual-queue update requires only $\mathcal O(M)$ operations.

Since the acceptance condition~\eqref{eq:acceptance_cond} is evaluated using the smoothed objective and unsuccessful backtracking trials retain the previous APV, the monotone-ascent property and the finite-horizon sample-path energy guarantee of Theorem~\ref{thm:energy} remain valid with finite-difference gradients. The stationary-point guarantee of Proposition~\ref{prop:convergence}, however, is established for exact gradients.

In single-process tests of an unoptimized Python implementation on an Apple M3 Max, OMP and LOMA together required 0.57~ms ($M=5$) and 0.65~ms ($M=8$) per slot on average. As a conservative timing check, the combined computation time was compared with the available data interval $T_s-\tau_t^{\mathrm{mov}}-\tau_{\mathrm p}$, which was exceeded in 0.30\% and 0\% of slots, respectively. A missed computation deadline can trigger a feasible stay-put action without violating the sample-path energy bound.

\section{Theoretical Analysis}
\label{sec:analysis}

We establish the convergence properties of the per-slot solver, derive finite-horizon and long-term movement-energy guarantees, and provide a state-matched throughput comparison that separates the effects of online energy allocation and approximate position optimization.

\subsection{Convergence of the Per-Slot Solver}
\label{subsec:solver_convergence}

We first establish the monotonicity and limiting stationarity of the proximal MM solver.

\begin{proposition}[Monotone ascent and stationarity]
\label{prop:convergence}
For each slot $t$, the inner-loop iterates of Algorithm~\ref{alg:prox_scamm} satisfy
\begin{equation}
  \tilde{F}\big(\mathbf{x}^{(i)}\big)
  \ge\tilde{F}\big(\mathbf{x}^{(i-1)}\big)
  +\tfrac{L_{\mathrm{init}}V}{2}
  \big\|\mathbf{x}^{(i)}-\mathbf{x}^{(i-1)}\big\|_{2}^{2},
  \label{eq:suff_ascent}
\end{equation}
for every executed iteration. If the inner loop is continued indefinitely, its objective values converge to a finite limit. Furthermore, under the exact-gradient regularity conditions of Section~\ref{subsec:mm}, if $J_{\max}\ge\lceil\log_{\eta}(L_{\mathrm{loc}}/(L_{\mathrm{init}}V))\rceil^{+}+1$ and the inner loop is continued indefinitely, every limit point $\bar{\mathbf{x}}$ is a stationary point of $\max_{\mathbf{x}\in\mathcal{A}}\tilde{F}(\mathbf{x})$, satisfying
\begin{equation}
  \mathbf{0}\in-\nabla\tilde{h}(\bar{\mathbf{x}})
  +\partial\psi(\bar{\mathbf{x}})
  +\mathcal{N}_{\mathcal{A}}(\bar{\mathbf{x}}),
  \label{eq:stationarity}
\end{equation}
where $\mathcal{N}_{\mathcal{A}}(\cdot)$ denotes the normal cone of $\mathcal{A}$.
\end{proposition}

\begin{IEEEproof}
Every accepted iterate is the $\mathbf{x}$-component of a solution to \eqref{eq:prox_qp}, while an unsuccessful backtracking search retains the previous feasible iterate. Hence, all iterates remain in $\mathcal{A}$. Inequality \eqref{eq:suff_ascent} follows from \eqref{eq:monotone} and $L\ge L_{\mathrm{init}}V$, with equality for a fallback step. Since $\mathcal{A}$ is compact and $\tilde{F}$ is continuous, the nondecreasing objective sequence is bounded above and converges if continued indefinitely. Summing \eqref{eq:suff_ascent} gives $\sum_i\|\mathbf{x}^{(i)}-\mathbf{x}^{(i-1)}\|_2^2<\infty$, implying $\|\mathbf{x}^{(i)}-\mathbf{x}^{(i-1)}\|_2\to0$.

Under the stated condition on $J_{\max}$ and with exact gradients, every iteration is accepted with $L^{(i)}\in[L_{\mathrm{init}}V,\bar L]$. The optimality condition of \eqref{eq:prox_problem} is
\begin{equation}
  \mathbf{0}\in
  L^{(i)}\big(\mathbf{x}^{(i)}-\mathbf{x}^{(i-1)}\big)
  -\nabla\tilde{h}\big(\mathbf{x}^{(i-1)}\big)
  +\partial\psi\big(\mathbf{x}^{(i)}\big)
  +\mathcal{N}_{\mathcal{A}}\big(\mathbf{x}^{(i)}\big).
  \label{eq:prox_opt}
\end{equation}
Let $\mathbf{x}^{(i_k)}\to\bar{\mathbf{x}}$ along a convergent subsequence. Since consecutive iterates become arbitrarily close, $\mathbf{x}^{(i_k-1)}\to\bar{\mathbf{x}}$ as well. The first term in \eqref{eq:prox_opt} vanishes because $L^{(i)}$ is bounded, while $\nabla\tilde h(\mathbf{x}^{(i_k-1)})\to\nabla\tilde h(\bar{\mathbf{x}})$ by continuity. Since $\psi$ is finite and convex, $\partial\psi+\mathcal{N}_{\mathcal{A}}=\partial(\psi+\iota_{\mathcal{A}})$, where $\iota_{\mathcal{A}}$ is the indicator function of $\mathcal{A}$. The subdifferential of this closed convex function has a closed graph~\cite{beck2017}. Passing to the limit establishes \eqref{eq:stationarity}.
\end{IEEEproof}

\subsection{Movement-Energy Guarantees}
\label{subsec:energy_guarantees}

We next establish a finite-horizon energy bound for every realized channel sequence with a well-defined per-slot rate, without requiring a uniform bound on channel gains across realizations.

\begin{theorem}[Energy guarantee]
\label{thm:energy}
For $T\ge1$, define $\Gamma_T\triangleq\max_{1\le t\le T}\hat\Gamma_t$, with $\Gamma_0\triangleq0$. Suppose that $0<E_{\mathrm{avg}}\le e_{\max}\triangleq c_{\mathrm{e}}MvT_d$ and $\epsilon+\mu\ln M\le vT_d$. For any $V>0$, Algorithm~\ref{alg:prox_scamm} with $Q_1=0$ satisfies
\begin{equation}
  Q_{T+1}\le Q_T^{\max}
  \triangleq\frac{V\Gamma_T}{E_{\mathrm{avg}}}+e_{\max},
  \label{eq:queue_bound}
\end{equation}
and consequently,
\begin{equation}
  \frac{1}{T}\sum_{t=1}^{T}e_t
  \le E_{\mathrm{avg}}+\frac{Q_T^{\max}}{T},
  \label{eq:energy_bound}
\end{equation}
for every finite horizon $T$ and every channel realization for which the per-slot objective is well defined.
\end{theorem}

\begin{IEEEproof}
The inner loop starts from the stay-put action $\mathbf{x}^{(0)}=\mathbf{x}_{t-1}$ and returns an iterate whose smoothed objective is no smaller. Using Proposition~\ref{prop:lse_bound}, we obtain
\begin{align}
F(\mathbf{x}_t)
&\ge\tilde{F}(\mathbf{x}_t)
 \ge\tilde{F}(\mathbf{x}_{t-1})\nonumber\\
&=\frac{V}{T_s}
  \left(T_d-\frac{\epsilon+\mu\ln M}{v}\right)
  \tilde{\gamma}^{\mathrm{sum}}
  (\mathbf{x}_{t-1};\bm{\xi}^{(t)})
  \ge0,
  \label{eq:F_nonneg}
\end{align}
where the equality uses $\Phi_{\mu,\epsilon}(\mathbf{0})=\epsilon+\mu\ln M$ and $\psi(\mathbf{x}_{t-1})=0$, and the last inequality follows from $\epsilon+\mu\ln M\le vT_d$. Since $F(\mathbf{x}_t)=VU_t-Q_te_t$ and $U_t\le\hat\Gamma_t$, it follows that $Q_te_t\le V\hat\Gamma_t$. Hence, whenever $Q_t>V\hat\Gamma_t/E_{\mathrm{avg}}$, we have $e_t<E_{\mathrm{avg}}$ and $Q_{t+1}<Q_t$.

We prove $Q_t\le V\Gamma_{t-1}/E_{\mathrm{avg}}+e_{\max}$ by induction. The claim holds for $t=1$ since $Q_1=0$. Suppose it holds at slot $t$. If $Q_t\le V\hat\Gamma_t/E_{\mathrm{avg}}$, then
\begin{equation}
  Q_{t+1}\le Q_t+e_{\max}
  \le\frac{V\Gamma_t}{E_{\mathrm{avg}}}+e_{\max}.
  \nonumber
\end{equation}
Otherwise, $Q_{t+1}<Q_t\le V\Gamma_{t-1}/E_{\mathrm{avg}}+e_{\max}\le V\Gamma_t/E_{\mathrm{avg}}+e_{\max}$. Thus, \eqref{eq:queue_bound} follows. Finally, the queue update implies $Q_{t+1}\ge Q_t+e_t-E_{\mathrm{avg}}$. Summing over $t=1,\ldots,T$ and using $Q_1=0$ gives $\sum_{t=1}^{T}(e_t-E_{\mathrm{avg}})\le Q_{T+1}\le Q_T^{\max}$, proving \eqref{eq:energy_bound}.
\end{IEEEproof}

Theorem~\ref{thm:energy} provides a sample-path bound without requiring global optimality of the per-slot solution and remains valid under early termination of the inner loop. For any channel realization satisfying $\Gamma_T/T\to0$, \eqref{eq:energy_bound} implies
\begin{equation}
  \limsup_{T\to\infty}\frac{1}{T}\sum_{t=1}^{T}e_t
  \le E_{\mathrm{avg}}.
  \nonumber
\end{equation}
In particular, consider the Gaussian fading model in \eqref{eq:gain_evolution} with Gaussian initial path gains and uniformly bounded large-scale gains. The NLoS path gains then have uniformly bounded variances and exponential tails in squared magnitude. A union bound and the Borel--Cantelli lemma imply that their running maximum squared magnitude grows at most as $\mathcal{O}(\ln T)$ almost surely. Together with the bounded LoS gains, \eqref{eq:realized_rate_bound} yields $\Gamma_T=\mathcal{O}(\ln\ln T)$ almost surely. Consequently, the long-term sample-path energy budget is satisfied almost surely. Since $0\le e_t\le e_{\max}$, the time-averaged energy is uniformly bounded, and the reverse Fatou lemma further gives
\begin{equation}
  \limsup_{T\to\infty}
  \frac{1}{T}\sum_{t=1}^{T}\mathbb{E}[e_t]
  \le E_{\mathrm{avg}}.
  \nonumber
\end{equation}
Thus, the expected energy constraint \eqref{eq:P1_energy} is also satisfied. When the channel gains are uniformly bounded, \eqref{eq:queue_bound} additionally provides a time-invariant backlog bound.

The condition $E_{\mathrm{avg}}\le e_{\max}$ is without loss of generality because a larger budget is never active. The parameter $V$ controls the energy transient: increasing $V$ permits a larger backlog and thus greater temporary expenditure above the average budget. Unlike a fixed movement-energy penalty tuned offline, the virtual queue adjusts the movement price according to realized energy consumption. Importantly, the energy bound relies on $F(\mathbf{x}_t)\ge0$ in \eqref{eq:F_nonneg}, not on monotonicity of the original nonsmoothed objective $F$.

\begin{remark}[Estimated CSI and scope of the energy bound]
\label{rem:estimated_energy}
Suppose a nonnegative estimated rate based on the available CSI is used consistently in the per-slot objective, while the virtual queue is updated using the movement energy defined in~\eqref{eq:energy}. If the returned iterate is no worse than the stay-put action in the estimated smoothed objective, the argument leading to \eqref{eq:F_nonneg} remains valid. The finite-horizon queue bound then holds with $\Gamma_T$ replaced by a finite running upper bound derived from the estimated channel. This does not automatically imply almost-sure long-term feasibility without an appropriate growth condition on that bound. Nor does estimated CSI alone establish the true-channel throughput comparison of Proposition~\ref{prop:reallocation}. Section~\ref{subsec:estimated_csi} therefore evaluates throughput and finite-horizon energy expenditure separately.
\end{remark}

\subsection{State-Matched Throughput Comparison}
\label{subsec:throughput_comparison}

We compare LOMA with reference actions evaluated along its realized antenna trajectory. Fix an integer frame length $R\ge1$ and define $\mathcal{T}_n\triangleq\{(n-1)R+1,\dots,nR\}$. For a reference sequence $\{\mathbf y_t\}$ with $\mathbf y_t\in\mathcal A(\mathbf x_{t-1})$, define $C_t\triangleq[F(\mathbf y_t)-F(\mathbf x_t)]^+$, where both objectives use LOMA's queue backlog $Q_t$. Since $F(\mathbf x_t)\ge0$ by \eqref{eq:F_nonneg} and $F(\mathbf y_t)\le V\hat\Gamma_t$, we have $0\le C_t\le V\hat\Gamma_t$.

\begin{proposition}[State-matched energy reallocation]
\label{prop:reallocation}
Consider any reference sequence $\{\mathbf y_t\}$ satisfying $\mathbf y_t\in\mathcal A(\mathbf x_{t-1})$ and
\begin{equation}
  \sum_{t\in\mathcal T_n}
  e_t(\mathbf y_t,\mathbf x_{t-1})
  \le R E_{\mathrm{avg}},
  \qquad \forall n.
  \nonumber
\end{equation}
The reference may depend on future channel states, but each action is evaluated from LOMA's actual preceding APV $\mathbf x_{t-1}$ rather than $\mathbf y_{t-1}$. Under the conditions of Theorem~\ref{thm:energy} and along any channel realization satisfying $\Gamma_T/T\to0$,
\begin{align}
  &\liminf_{T\to\infty}\frac{1}{T}
  \sum_{t=1}^{T}
  U_t(\mathbf x_t,\mathbf x_{t-1})
  \nonumber\\
  &\quad\ge
  \liminf_{T\to\infty}\frac{1}{T}
  \sum_{t=1}^{T}
  U_t(\mathbf y_t,\mathbf x_{t-1})
  \nonumber\\
  &\qquad
  -\frac{e_{\max}E_{\mathrm{avg}}+4RE_{\mathrm{avg}}^2}{V}
  -\bar c,
  \label{eq:reallocation}
\end{align}
where $\bar c\triangleq\limsup_{T\to\infty}\frac{1}{VT}\sum_{t=1}^{T}C_t$. The bound is vacuous if $\bar c=+\infty$, in which case its right-hand side is interpreted as $-\infty$.
\end{proposition}

\begin{IEEEproof}
Let $U_t$ and $e_t$ denote LOMA's utility and movement energy, and let $\tilde U_t$ and $\tilde e_t$ denote those of the reference action $\mathbf y_t$, both evaluated from $\mathbf x_{t-1}$. Since $0\le e_t\le e_{\max}$ and $E_{\mathrm{avg}}\le e_{\max}$, we have $|e_t-E_{\mathrm{avg}}|\le e_{\max}$. The queue update therefore gives
\begin{equation}
  \mathcal L(Q_{t+1})-\mathcal L(Q_t)
  \le \frac{e_{\max}}{2}|e_t-E_{\mathrm{avg}}|
  +Q_t(e_t-E_{\mathrm{avg}}).
  \nonumber
\end{equation}
By the definition of $C_t$, we also have $Q_te_t-VU_t\le Q_t\tilde e_t-V\tilde U_t+C_t$. Combining these inequalities and summing over $t=1,\dots,T$, with $T=NR$, yields
\begin{align}
  V\sum_{t=1}^{T}(\tilde U_t-U_t)
  \le\;&
  \frac{e_{\max}}{2}
  \sum_{t=1}^{T}|e_t-E_{\mathrm{avg}}|
  +\sum_{t=1}^{T}C_t
  \nonumber\\
  &+\sum_{n=1}^{N}
  \sum_{t\in\mathcal T_n}
  Q_t(\tilde e_t-E_{\mathrm{avg}}),
  \label{eq:realloc_sum}
\end{align}
where $\mathcal L(Q_1)=0$ and $\mathcal L(Q_{T+1})\ge0$ have been used.

Let $S\triangleq\sum_{t=1}^{T}(e_t+E_{\mathrm{avg}})$. By \eqref{eq:energy_bound}, $S\le2TE_{\mathrm{avg}}+Q_T^{\max}$, which also bounds $\sum_{t=1}^{T}|e_t-E_{\mathrm{avg}}|$. For each frame, let $t_n$ denote its first slot. Then
\begin{align}
  \sum_{t\in\mathcal T_n}
  Q_t(\tilde e_t-E_{\mathrm{avg}})
  =\;&
  Q_{t_n}
  \sum_{t\in\mathcal T_n}
  (\tilde e_t-E_{\mathrm{avg}})
  \nonumber\\
  &+
  \sum_{t\in\mathcal T_n}
  (Q_t-Q_{t_n})
  (\tilde e_t-E_{\mathrm{avg}}).
  \nonumber
\end{align}
The first term is nonpositive by the frame-wise budget constraint. Moreover, the queue update implies $|Q_{\tau+1}-Q_\tau|\le e_\tau+E_{\mathrm{avg}}$. Hence, for every $t\in\mathcal T_n$,
\begin{equation}
  |Q_t-Q_{t_n}|
  \le S_n\triangleq
  \sum_{\tau\in\mathcal T_n}
  (e_\tau+E_{\mathrm{avg}}).
  \nonumber
\end{equation}
The reference budget also gives $\sum_{t\in\mathcal T_n}|\tilde e_t-E_{\mathrm{avg}}|\le2RE_{\mathrm{avg}}$. Consequently,
\begin{equation}
  \sum_{t\in\mathcal T_n}
  Q_t(\tilde e_t-E_{\mathrm{avg}})
  \le2RE_{\mathrm{avg}}S_n.
  \nonumber
\end{equation}
Summing over frames and substituting into \eqref{eq:realloc_sum}, we obtain
\begin{align}
  \frac{1}{T}\sum_{t=1}^{T}(\tilde U_t-U_t)
  \le\;&
  \frac{e_{\max}E_{\mathrm{avg}}+4RE_{\mathrm{avg}}^2}{V}
  +\frac{1}{VT}\sum_{t=1}^{T}C_t
  \nonumber\\
  &+
  \left(\frac{e_{\max}}{2}+2RE_{\mathrm{avg}}\right)
  \frac{Q_T^{\max}}{VT}.
  \nonumber
\end{align}
Since $Q_T^{\max}/(VT)=\Gamma_T/(E_{\mathrm{avg}}T)+e_{\max}/(VT)\to0$, the last term vanishes. For an arbitrary horizon $T$, extending it to the nearest multiple of $R$ changes the normalized sums by at most $\mathcal{O}(R\Gamma_{T+R}/T)$, which vanishes for fixed $R$ under the assumed growth condition. Taking limits establishes \eqref{eq:reallocation}.
\end{IEEEproof}

For $R=1$, the reference must satisfy a per-slot energy budget, whereas $R>1$ permits energy reallocation within each frame. The state-matched comparison accounts for the dependence of feasible positions on previous decisions without assuming that the reference follows an independently evolving trajectory. The explicit $\mathcal{O}(R/V)$ term characterizes the energy-management contribution, while $\bar c$ captures the per-slot optimization residual. Section~\ref{sec:simulation} examines the resulting tradeoff between long-term and per-slot energy allocation.

\section{Simulation Results}
\label{sec:simulation}

We evaluate LOMA in terms of effective throughput, movement-energy consumption, and repositioning behavior. Comparisons with benchmark schemes examine how the benefit of long-term energy allocation depends on the temporal structure of channel variation. We also assess LOMA under pilot-based CSI acquisition.

\subsection{Simulation Setup}
\label{subsec:setup}

The BS employs equal-power ZF precoding with a uniformly spaced initial APV; other simulation parameters are summarized in Table~\ref{tab:params}. For numerical stability, ZF precoding is implemented by solving a linear system with a relative Tikhonov regularization term $\delta_{\mathrm{ZF}}\operatorname{tr}(\mathbf{H}^{\mathrm{H}}\mathbf{H})\mathbf{I}$, where $\delta_{\mathrm{ZF}}=10^{-12}$. The resulting precoder columns are normalized to unit norm. In the ideal-CSI experiments, $\tau_{\mathrm p}=0$, allowing each MA to move by at most $vT_s=0.03\lambda$ per slot. We set $c_{\mathrm e}=1$~J/m for the displacement-proportional movement-energy model. The average movement-energy budget is $E_{\mathrm{avg}}=2\times10^{-4}c_{\mathrm e}MvT_s$ (0.3~\textmu J/slot for $M=5$), independent of the pilot duration, while the per-slot movement energy is bounded by $e_{\max}=c_{\mathrm e}Mv(T_s-\tau_{\mathrm p})$. For smoothing, we use $\epsilon=0.05vT_s$ and $\mu=\epsilon/\ln M$, satisfying the condition of Theorem~\ref{thm:energy} for both pilot durations.

Sections~\ref{subsec:benchmarks}--\ref{subsec:perslot} assume ideal field-response CSI. Unless otherwise specified, results are averaged over the last 40,000 slots of each 80,000-slot realization to reduce the influence of transients associated with the initial APV and the zero-initialized virtual queue. We use 50 paired channel realizations for Tables~\ref{tab:main}--\ref{tab:environments} and Figs.~\ref{fig:v_sweep}--\ref{fig:eavg_sweep}. For the offline fixed-array benchmark, the one-time configuration cost is excluded from the recurring movement-energy consumption reported in the tables.

\begin{table}[t]
\caption{Simulation Parameters}
\label{tab:params}
\centering
\resizebox{\columnwidth}{!}{%
\renewcommand{\arraystretch}{1.1}
\begin{tabular}{clll}
\hline
& \textbf{Parameter} & \textbf{Symbol} & \textbf{Value} \\
\hline
\multirow{10}{*}{\rotatebox{90}{\textit{System}}}
& Number of MAs / UEs & $M$ / $K$ & $5$ / $5$ \\
& Carrier frequency (wavelength) & $f_{c}$ ($\lambda$) & $30$~GHz ($1$~cm) \\
& Movable region length & $D$ & $7.5\lambda$ ($=1.5M\lambda$) \\
& Minimum spacing & $d$ & $\lambda/2$ \\
& MA speed & $v$ & $0.3$~m/s \\
& Slot duration & $T_s$ & $1$~ms \\
& Pilot time (ideal / estimated CSI) & $\tau_{\mathrm p}$ & $0$ / $0.05$~ms \\
& Energy coefficient & $c_{\mathrm{e}}$ & $1$~J/m \\
& Transmit power & $P$ & $30$~dBm \\
& Noise power & $\sigma_{n}^{2}$ & $-90$~dBm \\
\hline
\multirow{9}{*}{\rotatebox{90}{\textit{Channel}}}
& Reference path gain & $\beta_{0}$ & $-40$~dB at $1$~m \\
& Path-loss exponent & -- & $2.5$ \\
& Number of NLoS paths & $L_{k}$ & $3$ \\
& Rician factor & $\kappa$ & $10$~dB \\
& NLoS AoDs & $\theta^{\mathrm{NLoS}}_{k,\ell}$ & $\mathcal{U}[15^{\circ},165^{\circ}]$ \\
& UE distance / angle & $r^{(0)}_{k}$ / -- & $\mathcal{U}[30,100]$~m / $\mathcal{U}[30^{\circ},150^{\circ}]$ \\
& Mean event interval (Env.~I) & -- & $3000$ slots \\
& UE speed (Env.~II) & $v_{\mathrm{UE}}$ & $0.2$~m/s \\
& Heading innovation std. (Env.~II) & $\sigma_{\phi}$ & $0.5^{\circ}$ per slot \\
\hline
\multirow{5}{*}{\rotatebox{90}{\textit{Algorithm}}}
& Control parameter & $V$ & $3\times10^{-6}$ \\
& Energy budget & $E_{\mathrm{avg}}$ & $2\times10^{-4}c_{\mathrm e}MvT_s$ \\
& Initial Lipschitz scale & $L_{\mathrm{init}}$ & $10^{5}$ \\
& Backtracking factor / cap & $\eta$ / $J_{\max}$ & $2$ / $20$ \\
& Tolerance / max. iterations & $\delta$ / $I_{\max}$ & $10^{-4}$ / $100$ \\
\hline
\end{tabular}%
}
\end{table}

We consider a deployment with stationary UEs, such as fixed wireless access or indoor machine-type communications. With $v_{\mathrm{UE}}=0$ in~\eqref{eq:ue_mobility}, the propagation geometry of each link remains unchanged while a UE is served. Channel variation therefore arises solely from the \emph{event-driven variation} described in Section~\ref{subsec:temporal}, leading to temporally concentrated repositioning demand. Reconfiguration events follow a renewal process with a mean inter-event interval of approximately $3000$ slots. This interval is substantially longer than the $250$ slots corresponding to a full-aperture displacement at the actuator speed, allowing substantial repositioning between typical events. At each event, the scheduled user set changes through user turnover or rescheduling, modeled by independently redrawing the field-response parameters $\bm{\xi}^{(t)}_{k}$ of one UE, including its distance, LoS AoD, and scattering geometry. Section~\ref{subsec:perslot} additionally considers continuous channel variation caused by UE mobility.

We consider four benchmark schemes to distinguish the benefits of antenna repositioning and temporal energy allocation. The \emph{FPA} benchmark fixes the uniformly spaced initial APV. The \emph{optimized fixed} (Opt. fixed) benchmark performs projected gradient ascent from 20 initial APVs (one uniformly spaced and 19 random) and then fixes the candidate with the highest first-slot sum rate. The \emph{myopic} scheme maximizes effective throughput in each slot with $Q_t\equiv0$ and no energy constraint (Remark~\ref{rem:myopic}). Finally, \emph{LOMA-PS} imposes a strict per-slot actuation-energy budget using the same position solver, as detailed in Section~\ref{subsec:perslot}. All moving schemes obey the same inter-slot reachability constraints, and all methods are evaluated on paired channel realizations.

\begin{table}[t]
\caption{Comparison of effective throughput, normalized movement energy, and average number of actuated MAs per slot in Environment~I under ideal CSI, with $D=1.5M\lambda$ and $50$ paired channel realizations.}
\label{tab:main}
\centering
\renewcommand{\arraystretch}{1.15}
\begin{tabular}{llccc}
\hline
\multirow{2}{*}{$M=K$} & \multirow{2}{*}{Scheme} & Throughput & Energy & Actuations \\
 & & [bit/s/Hz] & / $E_{\mathrm{avg}}$ & per slot \\
 \hline
 \multirow{4}{*}{$4$}
  & FPA        & $34.412$ & $0.000$ & $0.000$ \\
  & Opt. fixed & $35.531$ & $0.000$ & $0.000$ \\
  & LOMA       & $\mathbf{39.418}$ & $1.025$\pmse{0.031} & $1.277$ \\
  & Myopic     & $40.947$ & $2.961$\pmse{0.086} & $3.673$ \\
\hline
\multirow{4}{*}{$5$}
  & FPA        & $39.390$ & $0.000$ & $0.000$ \\
  & Opt. fixed & $42.133$ & $0.000$ & $0.000$ \\
  & LOMA       & $\mathbf{47.183}$ & $0.945$\pmse{0.026} & $1.442$ \\
  & Myopic     & $49.568$ & $2.911$\pmse{0.093} & $4.575$ \\
\hline
\multirow{4}{*}{$8$}
  & FPA        & $57.676$ & $0.000$ & $0.000$ \\
  & Opt. fixed & $62.428$ & $0.000$ & $0.000$ \\
  & LOMA       & $\mathbf{72.837}$ & $1.025$\pmse{0.023} & $2.250$ \\
  & Myopic     & $77.139$ & $3.098$\pmse{0.079} & $7.621$ \\
\hline
\end{tabular}
\end{table}

\begin{figure*}[t]
\centering
\subfloat[Effective throughput.]{
  \includegraphics[width=0.45\textwidth]{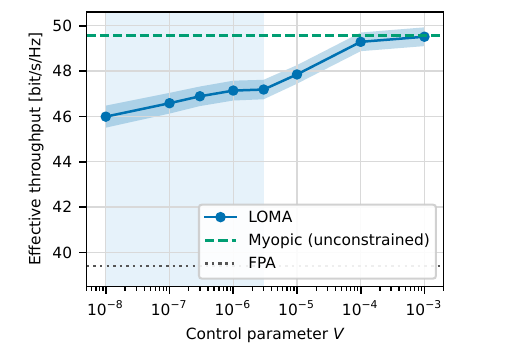}
  \label{fig:v_a}}
\hfill
\subfloat[Normalized energy (left axis) and actuations per slot (right axis).]{
  \includegraphics[width=0.45\textwidth]{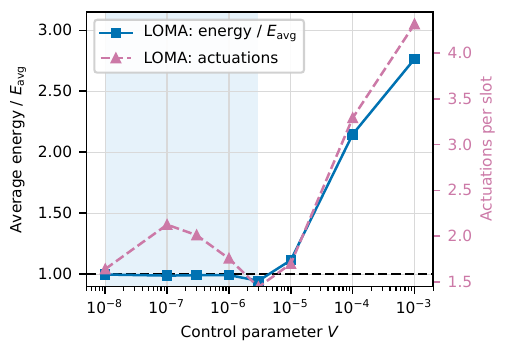}
  \label{fig:v_b}}
\caption{Effect of the control parameter $V$ on effective throughput, movement energy, and antenna actuation. Shaded bands indicate one standard error, and the highlighted region denotes the values of $V$ for which the observed average movement energy remains within the budget.}
\label{fig:v_sweep}
\end{figure*}

\subsection{Comparison with Benchmark Schemes}\label{subsec:benchmarks}
\label{subsec:comparison}

Table~\ref{tab:main} examines whether online antenna repositioning remains beneficial under a stringent average movement-energy budget and how much of the unconstrained repositioning gain can be retained. Across all three array sizes, LOMA consistently outperforms both fixed-array benchmarks. For $M=5$, LOMA achieves 47.183~bit/s/Hz, a 12.0\% improvement over the optimized fixed array. More importantly, it recovers 76.6\% of the FPA-to-myopic throughput gain while consuming only 32.5\% of the myopic scheme's movement energy. This trend is consistent across $M=4,5,8$, with approximately 77\% of the throughput gain recovered using about one-third of the movement energy. LOMA operates near the prescribed average energy budget, whereas the myopic scheme consumes approximately three times the budget.

The energy savings are accompanied by substantially fewer antenna updates. For $M=5$, LOMA moves an average of 1.442 antennas per slot, compared with 4.575 for the myopic scheme, consistent with the queue-dependent soft-thresholding mechanism in Remark~\ref{rem:sparse}. Under more frequent user reconfigurations (mean interval of 1000 slots), LOMA still recovers 65.1\% of the FPA-to-myopic throughput gain (20 paired realizations). Additional tests with RZF and MRT show that both LOMA modes outperform the corresponding fixed-array baselines in both environments (20 paired realizations; all significant at the 1\% level), indicating that the benefit of repositioning is not specific to ZF precoding. The advantage of allocating energy across slots, rather than enforcing a per-slot budget, is examined separately in Section~\ref{subsec:perslot}.

\begin{figure*}[t]
\centering
\subfloat[Effective throughput.]{
  \includegraphics[width=0.45\textwidth]{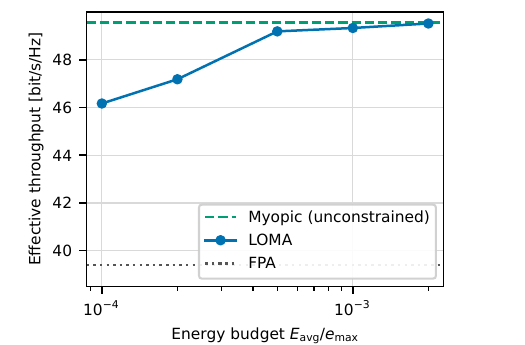}
  \label{fig:eavg_a}}
\hfill
\subfloat[Fraction of the FPA-to-myopic gain recovered by LOMA.]{
  \includegraphics[width=0.45\textwidth]{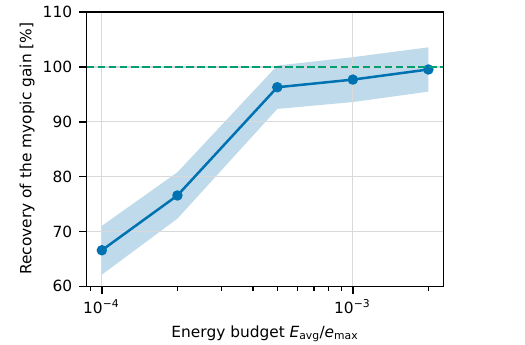}
  \label{fig:eavg_b}}
\caption{Effect of the average movement-energy budget on effective throughput and the fraction of the FPA-to-myopic gain recovered by LOMA. Shaded bands indicate one standard error; $0\%$ and $100\%$ recovery correspond to the FPA and myopic throughput levels, respectively.}
\label{fig:eavg_sweep}
\end{figure*}

\begin{figure*}[t]
\centering
\subfloat[Normalized per-slot movement energy (left axis) and virtual queue (right axis).]{
  \includegraphics[width=0.45\textwidth]{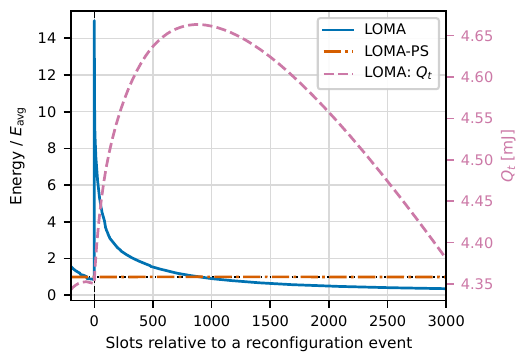}\label{fig:traj_a}}
\hfill
\subfloat[Effective throughput.]{
  \includegraphics[width=0.45\textwidth]{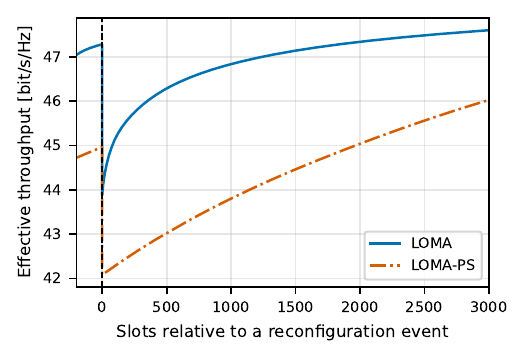}\label{fig:traj_b}}
\caption{Transient response to user-reconfiguration events in Environment~I. The curves are aligned at the event slot ($0$) and averaged over 127 isolated events from 20 channel realizations, covering 200 slots before and 3000 slots after each event.}
\label{fig:traj}
\end{figure*}

\subsection{Effect of the Design Parameters}
\label{subsec:sweeps}

Fig.~\ref{fig:v_sweep} examines how the control parameter $V$ affects the trade-off between throughput and movement-energy consumption. Increasing $V$ improves throughput by placing greater emphasis on the rate relative to the queue penalty, but excessively large values lead to higher energy consumption. At the selected operating point $V=3\times10^{-6}$, LOMA achieves 47.183~bit/s/Hz while consuming 94.5\% of the average energy budget. Further increasing $V$ brings the throughput closer to the myopic level but exceeds the budget over the evaluation horizon. The number of actuations varies nonmonotonically within the budget-feasible region, indicating that energy consumption depends on displacement magnitudes as well as actuation frequency. These observations are consistent with Theorem~\ref{thm:energy}, whose finite-horizon energy bound allows a larger transient deviation as $V$ increases.

Fig.~\ref{fig:eavg_sweep} evaluates how much of the myopic throughput gain can be recovered under different movement-energy budgets. Even at the tightest budget $E_{\mathrm{avg}}/e_{\max}=10^{-4}$, LOMA recovers approximately 67\% of the FPA-to-myopic throughput gain. This fraction increases to approximately 77\% at the default budget
$2\times10^{-4}$ and 96\% at $5\times10^{-4}$. As the budget is relaxed further, LOMA approaches the myopic throughput level, with diminishing additional gains. These results demonstrate that allocating a limited movement-energy budget across slots can retain a substantial portion of the benefits of unconstrained repositioning.

\subsection{Long-Term Versus Per-Slot Budget}
\label{subsec:perslot}

We examine when allocating movement energy across slots is more beneficial than enforcing a strict per-slot budget. For this comparison, LOMA-PS employs the same proximal-QP framework as LOMA but imposes $e_t\le E_{\mathrm{avg}}$ in every slot. Specifically, the queue penalty is removed, the constraint $\mathbf{1}^{\mathrm T}\mathbf{s}\le E_{\mathrm{avg}}/c_{\mathrm e}$ is added to \eqref{eq:prox_qp}, and the virtual-queue update is skipped. We compare two channel environments: Environment~I features event-driven user reconfigurations, whereas Environment~II exhibits continuous channel variation due to UE mobility at $v_{\mathrm{UE}}=0.2$~m/s. For Environment~II, we set $V=3\times10^{-7}$ to keep LOMA's average movement energy near the prescribed budget.

Fig.~\ref{fig:traj} illustrates why a long-term budget is advantageous under event-driven channel variation. Immediately after a user reconfiguration, LOMA temporarily consumes more than ten times the nominal per-slot energy budget, allowing rapid throughput recovery. Its virtual queue rises from approximately 4.35 to 4.66~mJ and subsequently decreases toward 4.38~mJ as movement energy consumption falls below the budget. In contrast, LOMA-PS cannot concentrate energy around the event and recovers throughput more gradually. Consequently, LOMA achieves 47.183~bit/s/Hz in Environment~I, exceeding LOMA-PS by $1.17\pm0.28$~bit/s/Hz (50 paired realizations; Table~\ref{tab:environments}). Since both schemes use the same position-optimization framework, this improvement demonstrates the benefit of distributing movement energy unevenly across slots.

The ordering reverses under continuous channel variation. As shown in Table~\ref{tab:environments}, LOMA-PS achieves 44.790~bit/s/Hz in Environment~II, compared with 42.967~bit/s/Hz for LOMA, while both operate near the average energy budget. LOMA-PS performs 2.094 actuations per slot, whereas LOMA performs only 0.169, indicating that frequent small adjustments are more effective for tracking continuously varying channels than less frequent, larger repositioning actions. Together, these results show that the advantage of long-term energy allocation depends on the temporal structure of channel variation: concentrated repositioning is beneficial after intermittent reconfigurations, whereas incremental adjustments can be preferable under persistent channel drift.

\begin{table}[t]
\caption{Comparison under event-driven (Environment~I) and continuous (Environment~II) channel variation with $M=K=5$ and ideal CSI; $50$ paired realizations.}
\label{tab:environments}
\centering
\renewcommand{\arraystretch}{1.15}
\begin{tabular}{llccc}
\hline
\multirow{2}{*}{Env.} & \multirow{2}{*}{Scheme} & Throughput & Energy & Actuations \\
 & & [bit/s/Hz] & / $E_{\mathrm{avg}}$ & per slot \\
\hline
\multirow{4}{*}{I}
  & FPA        & $39.390$ & $0.000$ & $0.000$ \\
  & Opt. fixed & $42.133$ & $0.000$ & $0.000$ \\
  & LOMA-PS    & $46.018$ & $0.978$\pmse{0.008} & $2.322$ \\
  & LOMA       & $\mathbf{47.183}$ & $0.945$\pmse{0.026} & $1.442$ \\
\hline
\multirow{4}{*}{II}
  & FPA        & $38.778$ & $0.000$ & $0.000$ \\
  & Opt. fixed & $42.486$ & $0.000$ & $0.000$ \\
  & LOMA-PS    & $\mathbf{44.790}$ & $0.993$\pmse{0.002} & $2.094$ \\
  & LOMA       & $42.967$ & $0.954$\pmse{0.040} & $0.169$ \\
\hline
\end{tabular}
\end{table}

\begin{table}[t]
\caption{Effect of CSI quality on effective throughput and normalized movement energy in Environment~I ($K=5$, $\tau_{\mathrm p}=0.05$~ms, 20 paired realizations). All schemes start from the uniform APV; CSI pairs specify positioning/precoding information.}
\label{tab:pilot_csi}
\centering
\renewcommand{\arraystretch}{1.12}
\begin{tabular}{lrr}
\toprule
Scheme and CSI setting & $M=5$ & $M=8$ \\
\midrule
FPA, ideal precoder CSI & 37.04 & 51.16 \\
FPA, LS precoder CSI & 33.48 & 47.37 \\
LOMA, ideal / ideal & 44.67 & 54.15 \\
LOMA, OMP / ideal & 43.41 & 53.70 \\
LOMA, OMP / LS & 39.71 & 49.89 \\
LOMA-PS, OMP / LS & 38.84 & 49.71 \\
\midrule
\multicolumn{3}{l}{\emph{LOMA (OMP / LS): movement energy / $E_{\mathrm{avg}}$}} \\
Last 40,000 slots & 1.006 & 0.935 \\
All 80,000 slots & 1.035 & 0.955 \\
\bottomrule
\end{tabular}
\end{table}

\subsection{Pilot-Based Operation With Estimated CSI}
\label{subsec:estimated_csi}

We examine whether online repositioning remains beneficial when CSI is estimated from ordinary TDD pilots without mechanically probing candidate antenna positions. After the MAs settle at $\mathbf{x}_t$, a single pilot observation provides the LS channel estimate for slot-$t$ ZF precoding. OMP processes the same observation to construct a sparse field-response model for selecting $\mathbf{x}_{t+1}$. Accordingly, line~3 of Algorithm~\ref{alg:prox_scamm} uses the OMP estimate obtained in slot $t-1$, while precoding uses the current-position LS estimate. This allows LOMA to evaluate candidate APVs from the estimated field response without moving the antennas to measure their channels. No additional mechanical movements are required solely for CSI acquisition. For each UE, OMP estimates four propagation paths using a 241-point uniform grid in $u=\cos\theta$ over $[\cos165^\circ,\cos15^\circ]$. Each selected direction is refined on a 41-point grid spanning one coarse-grid step on either side, with path gains re-estimated by LS after each selection.

We consider Environment~I with $K=5$ and $M\in\{5,8\}$, where the aperture is $D=1.5M\lambda$. Each UE transmits five orthogonal pilot symbols at 23~dBm, yielding a median antenna SNR of 34.6~dB. The pilot duration is $\tau_{\mathrm p}=0.05$~ms, and the absolute energy budget remains $E_{\mathrm{avg}}=2\times10^{-4}c_{\mathrm e}MvT_s$. We set $V=3\times10^{-7}$ for $M=5$ and $V=3\times10^{-6}$ for $M=8$, selected using five calibration realizations independent of the 20 paired evaluation realizations. All schemes start from the same uniform APV. To isolate the effects of CSI estimation, we compare ideal/ideal, OMP/ideal, and OMP/LS CSI for positioning and precoding under otherwise identical pilot settings.

Table~\ref{tab:pilot_csi} demonstrates that online repositioning retains substantial throughput gains under fully estimated CSI. With OMP-based positioning and LS precoding, LOMA achieves 39.71~bit/s/Hz for $M=5$, exceeding the corresponding LS-based FPA by $6.23\pm0.82$~bit/s/Hz (18.6\%). For $M=8$, the gain is $2.51\pm0.54$~bit/s/Hz (5.3\%). The reported uncertainties are paired-difference standard errors. LOMA-PS also improves upon the same fixed-array benchmark by 16.0\% and 4.9\%, respectively. These comparisons use identical CSI conditions, confirming that the gains arise from online repositioning rather than more accurate precoding CSI. Relative to the ideal-CSI gains over the corresponding FPA benchmarks, LOMA with OMP/LS CSI retains approximately 82\% and 84\% of the repositioning gains for $M=5$ and $M=8$, respectively.

The controlled CSI comparisons further identify the impact of each estimation stage. Replacing ideal field-response CSI with OMP estimates for position selection reduces throughput by $1.26\pm0.48$ and $0.45\pm0.15$~bit/s/Hz for $M=5$ and $M=8$, respectively. Subsequently replacing ideal precoding CSI with LS estimates causes additional losses of 3.70 and 3.81~bit/s/Hz. These losses are close to the corresponding FPA precoding losses of 3.56 and 3.79~bit/s/Hz. Thus, under the tested sparse channels, the throughput penalty associated with field-response estimation for repositioning is substantially smaller than that caused by imperfect CSI for ZF precoding.

We also evaluate a stronger fixed-array benchmark with a nonuniform APV designed using channel realizations independent of those used for evaluation. Despite starting from the uniform APV, LOMA with OMP/LS CSI outperforms this designed fixed array by $3.24\pm0.79$~bit/s/Hz (8.9\%) for $M=5$ and $0.60\pm0.20$~bit/s/Hz (1.2\%) for $M=8$. Optimizing LOMA's initial APV does not significantly change its long-run throughput, whereas offline position optimization improves the fixed-array benchmark. This indicates that online repositioning reduces the importance of the initial antenna configuration.

LOMA also operates near the prescribed long-term energy budget. Its normalized movement energy over the last 40,000 slots is $1.006\pm0.022$ for $M=5$ and $0.935\pm0.021$ for $M=8$, with corresponding full-horizon averages of 1.035 and 0.955. In contrast, LOMA-PS satisfies its per-slot energy constraint by construction.

Finally, we examine the effect of closely spaced AoDs on local channel prediction. In a separate known-AoD LS diagnostic for the uniform $M=5$ array at the operating pilot SNR, the estimation error of individual path gains deteriorates from approximately $-30$ to $-13.8$~dB as the minimum AoD separation decreases. Nevertheless, the combined one-slot channel-prediction error remains near $-34$~dB ($-34.2$ to $-33.5$~dB) across the tested separation groups. This illustrates that accurate prediction of the combined channel response does not necessarily require accurate recovery of each individual path gain. The end-to-end OMP/LS results above further demonstrate the benefit of repositioning when the path directions must also be estimated.

Overall, these results support LOMA in sparse, LoS-dominant mmWave deployments with relatively stable propagation geometry and intermittent user reconfigurations. In such environments, field-response estimation from pilots at the settled antenna positions enables effective online repositioning without dedicated mechanical CSI-probing movements. Computational complexity and runtime measurements are reported in Section~\ref{sec:algorithm}.

\section{Conclusion}
\label{sec:conclusion}

We investigated online MA repositioning under movement delay, inter-slot reachability, and a long-term actuation-energy budget. We developed LOMA, which combines virtual-queue energy pricing with proximal MM position updates, and established monotone per-slot improvement, a finite-horizon sample-path energy bound, conditional almost-sure long-term feasibility, and a state-matched throughput comparison. Simulations demonstrated that allocating movement energy across slots improves effective throughput under intermittent user reconfigurations, whereas a strict per-slot budget can better support continuous channel tracking. In sparse LoS-dominant mmWave channels, OMP-based field-response estimation using ordinary TDD pilots preserved throughput gains over uniform and offline-designed fixed arrays without dedicated mechanical CSI-probing movements. Future work will extend the framework to faster channel dynamics and investigate the joint design of channel acquisition and antenna repositioning.

\bibliographystyle{IEEEtran}
\bibliography{LOMA_ref}

@book{bertsekas1999,
  author    = {Dimitri P. Bertsekas},
  title     = {Nonlinear Programming},
  edition   = {2nd},
  publisher = {Athena Scientific},
  address   = {Belmont, MA},
  year      = {1999}
}

@article{zhu2023movable,
  title={Movable antennas for wireless communication: Opportunities and challenges},
  author={Zhu, Lipeng and Ma, Wenyan and Zhang, Rui},
  journal = {IEEE Commun. Mag.},
  volume={62},
  number={6},
  pages={114--120},
  year={2024},
  month={Jun.},
  publisher={IEEE}
}

@article{ding2025energy,
  title={Energy efficiency maximization for movable antenna communication systems},
  author={Ding, Jingze and Zhou, Zijian and Zhu, Lipeng and Zhao, Yuping and Jiao, Bingli and Zhang, Rui},
  journal={IEEE Trans. Wireless Commun.},
  volume={25},
  pages={2624--2638},
  year={2026},
  publisher={IEEE}
}

@article{xiao2024channel,
  title={Channel estimation for movable antenna communication systems: A framework based on compressed sensing},
  author={Xiao, Zhenyu and Cao, Songqi and Zhu, Lipeng and Liu, Yanming and Ning, Boyu and Xia, Xiang-Gen and Zhang, Rui},
  journal={IEEE Trans. Wireless Commun.},
  volume={23},
  number={9},
  pages={11814--11830},
  month={Sep.},
  year={2024},
  publisher={IEEE}
}

@inproceedings{wei2025mechanical,
  title     = {Mechanical Power Modeling and Energy Efficiency Maximization for Movable Antenna Systems},
  author    = {Wei, Xin and Mei, Weidong and Huang, Xuan and Chen, Zhi and Ning, Boyu},
  booktitle = {Proc. IEEE Global Commun. Conf. (GLOBECOM)},
  address   = {Taipei, Taiwan},
  month     = dec,
  year      = {2025},
  pages     = {6117--6122}
}

@article{wei2026energy,
  title={Energy-efficient movable antennas: Mechanical power modeling and performance optimization},
  author={Wei, Xin and Mei, Weidong and Huang, Xuan and Chen, Zhi and Ning, Boyu},
  journal={IEEE Trans. Wireless Commun.},
  volume={25},
  pages={16888--16902},
  year={2026},
  publisher={IEEE}
}

@article{wang2025throughput,
  title={Throughput maximization for movable antenna systems with movement delay consideration},
  author={Wang, Honghao and Wu, Qingqing and Gao, Ying and Chen, Wen and Mei, Weidong and Hu, Guojie and Xu, Lexi},
  journal={IEEE Trans. Wireless Commun.},
  volume={25},
  pages={883--899},
  year={2026},
  publisher={IEEE}
}

@article{li2025trajectory,
  title={Trajectory optimization for minimizing movement delay in movable antenna systems},
  author={Li, Qingliang and Mei, Weidong and Zhang, Rui and Ning, Boyu},
  journal={IEEE Trans. Wireless Commun.},
  volume={25},
  pages={6986--6999},
  year={2026},
  publisher={IEEE}
}

@article{feng2024weighted,
  title={Weighted sum-rate maximization for movable antenna-enhanced wireless networks},
  author={Feng, Biqian and Wu, Yongpeng and Xia, Xiang-Gen and Xiao, Chengshan},
  journal={IEEE Wireless Commun. Lett.},
  volume={13},
  number={6},
  pages={1770--1774},
  month={Jun.},
  year={2024},
  publisher={IEEE}
}

@article{ning2025movable,
  title={Movable antenna-enhanced wireless communications: General architectures and implementation methods},
  author={Ning, Boyu and Yang, Songjie and Wu, Yafei and Wang, Peilan and Mei, Weidong and Yuen, Chau and Bj{\"o}rnson, Emil},
  journal={IEEE Wireless Commun.},
  volume={32},
  number={5},
  pages={108--116},
  month={Oct.},
  year={2025},
  publisher={IEEE}
}

@article{ma2024multi,
  title={Multi-beam forming with movable-antenna array},
  author={Ma, Wenyan and Zhu, Lipeng and Zhang, Rui},
  journal={IEEE Commun. Lett.},
  volume={28},
  number={3},
  pages={697--701},
  month={Mar.},
  year={2024},
  publisher={IEEE}
}

@article{zheng2025two,
  author  = {Zheng, Z. and Wu, Q. and Chen, W. and Hu, G.},
  title   = {Two-Timescale Design for Movable Antenna-Enabled Multiuser {MIMO} Systems},
  journal = {IEEE Trans. Commun.},
  volume  = {73},
  number  = {11},
  pages   = {10554--10571},
  month   = {Nov.},
  year    = {2025}
}

@book{beck2017,
  author    = {A. Beck},
  title     = {First-Order Methods in Optimization},
  publisher = {SIAM},
  address   = {Philadelphia, PA},
  year      = {2017}
}

@book{neely2010,
  author    = {Michael J. Neely},
  title     = {Stochastic Network Optimization with Application to Communication and Queueing Systems},
  publisher = {Morgan \& Claypool},
  series    = {Synthesis Lectures on Communication Networks},
  year      = {2010}
}

@article{nesterov2005,
  author  = {Yurii Nesterov},
  title   = {Smooth Minimization of Non-Smooth Functions},
  journal = {Mathematical Programming},
  volume  = {103},
  number  = {1},
  pages   = {127--152},
  year    = {2005}
}

@book{boyd2004,
  author    = {Stephen Boyd and Lieven Vandenberghe},
  title     = {Convex Optimization},
  publisher = {Cambridge University Press},
  year      = {2004}
}

@article{zheng2026trajectory,
  author = {Zheng, Xi and Zeng, Cheng and Xu, Jie and Zhang, Rui},
  title = {Channel Estimation and Trajectory Design for Movable Antenna-Aided Communication With Time Constraint},
  journal = {IEEE Trans. Wireless Commun.},
  volume = {25},
  pages = {18199--18214},
  year = {2026},
  doi = {10.1109/TWC.2026.3699090}
}

@article{xu2024fas,
  author = {Xu, Hao and Zhou, Gui and Wong, Kai-Kit and New, Wee Kiat and Wang, Chao and Chae, Chan-Byoung and Murch, Ross and Jin, Shi and Zhang, Yangyang},
  title = {Channel Estimation for {FAS}-Assisted Multiuser {mmWave} Systems},
  journal = {IEEE Commun. Lett.},
  volume = {28},
  number = {3},
  pages = {632--636},
  year = {2024},
  month = mar,
  doi = {10.1109/LCOMM.2023.3347951}
}

@misc{ch25comove,
  author = {Choi, Wonseok and Tae, Yunseob and Lee, Jeongjae and Hong, Songnam},
  title = {{CSI}-Free Positioning of Movable Antennas for {IoT} Networks: A Compositional Kernelized Bandit},
  year = {2026},
  eprint = {2610.01285},
  archivePrefix = {arXiv},
  primaryClass = {eess.SP},
  url = {https://arxiv.org/abs/2610.01285}
}

@misc{ch26move,
  author = {Choi, Wonseok and Tae, Yunseob and Lee, Jeongjae and Hong, Songnam},
  title = {Positioning a Movable Antenna Without {CSI}: A Kernelized Bandit Under Costly Movement},
  year = {2026},
  eprint = {2609.27226},
  archivePrefix = {arXiv},
  primaryClass = {eess.SP},
  url = {https://arxiv.org/abs/2609.27226}
}

@inproceedings{tae2026vtc,
  author = {Tae, Yunseob and Lee, Jeongjae and Hong, Songnam},
  title = {Movement Delay-Aware Sum-Rate Maximization for Movable Antenna Systems},
  booktitle = {Proc. IEEE Vehicular Technology Conference (VTC-Fall)},
  address = {Boston, MA, USA},
  month = sep,
  year = {2026}
}

@misc{zhang2026relay,
  author = {Zhang, Yuanjun and Zhu, Lipeng and Mao, Haobin and Li, Zhe and Xiao, Zhenyu},
  title = {Movable-Antenna Relay Enhanced Multiuser Communications: Reinforcement Learning-Based Antenna Positioning and Power Allocation},
  year = {2026},
  eprint = {2610.20142},
  archivePrefix = {arXiv},
  primaryClass = {eess.SP},
  url = {https://arxiv.org/abs/2610.20142}
}

@article{liang2026twotimescale,
  author = {Liang, Fengyu and Cai, Yunlong and Liu, An and Champagne, Benoit},
  title = {Two-Timescale Deep Optimization of Positioning and Beamforming in Movable Antenna Arrays},
  journal = {IEEE Trans. Wireless Commun.},
  volume = {25},
  pages = {14242--14258},
  year = {2026},
  doi = {10.1109/TWC.2026.3675048}
}

@article{deng2026dualtimescale,
  author = {Deng, Z. and Wu, T. and Lyu, B. and Liu, Y. and Wu, M. and Nallanathan, A. and Jamalipour, A.},
  title = {Joint Design of Dual-Timescale Position, User Association and Transceiver Beamforming for Movable-Antenna Enabled Heterogeneous Networks},
  journal = {IEEE Trans. Veh. Technol.},
  year = {2026},
  doi = {10.1109/TVT.2026.3730533},
  note = {early access}
}
\end{document}